\documentclass[12pt, draftclsnofoot, onecolumn]{IEEEtran}
\usepackage{epsfig,latexsym}
\usepackage{float}
\usepackage{indentfirst}
\usepackage{amsmath}
\usepackage{bm}
\usepackage{amssymb}
\usepackage{times}

\usepackage{algorithm}
\usepackage[noend]{algpseudocode}
\usepackage{subfigure}
\usepackage{psfrag}
\usepackage{hyperref}
\usepackage{cite}
\usepackage{lastpage}
\usepackage{fancyhdr}
\usepackage{color}
 \usepackage{amsthm}
\usepackage{bigints}
\newtheorem{Lemma}{Lemma}
\newtheorem{Corollary}{Corollary}
\newtheorem{lemma}[Lemma]{$\mathbf{Lemma}$}

\newtheorem{corollary}[Corollary]{$\mathbf{Corollary}$}

\newcounter{problem}
\newcounter{save@equation}
\newcounter{save@problem}
\makeatletter

\begin{document}
\title{\vspace{-0.5em} \huge{ Age of Information: Can CR-NOMA Help?   }}

\author{ Zhiguo Ding, \IEEEmembership{Fellow, IEEE},   Robert Schober, \IEEEmembership{Fellow, IEEE}, and H. Vincent Poor, \IEEEmembership{Life Fellow, IEEE}    \thanks{ 
  
\vspace{-2em}

    Z. Ding and H. V. Poor are  with the Department of
Electrical and Computer Engineering, Princeton University, Princeton, NJ 08544,
USA. Z. Ding
 is also  with the School of
Electrical and Electronic Engineering, the University of Manchester, Manchester, UK (email: \href{mailto:zhiguo.ding@manchester.ac.uk}{zhiguo.ding@manchester.ac.uk}, \href{mailto:poor@princeton.edu}{poor@princeton.edu}).
  R. Schober is with the Institute for Digital Communications,
Friedrich-Alexander-University Erlangen-Nurnberg (FAU), Germany (email: \href{mailto:robert.schober@fau.de}{robert.schober@fau.de}).
 

  }\vspace{-2em}}
 \maketitle

\vspace{-1em}
\begin{abstract}
The aim of this paper is to  exploit    cognitive-ratio inspired NOMA (CR-NOMA) transmission  to reduce   the age of information in  wireless networks.  In particular, two CR-NOMA transmission protocols are developed by   utilizing the key features of different data generation models and applying CR-NOMA as an add-on to a legacy orthogonal multiple access (OMA) based network.   The fact that   the implementation of CR-NOMA causes little   disruption to the legacy OMA network means that  the proposed CR-NOMA protocols can be practically   implemented in   various communication systems which are based on OMA.  Closed-form expressions for  the AoI achieved by the proposed NOMA protocols are developed to facilitate performance evaluation, and asymptotic studies are carried out to identify two benefits of using   NOMA to reduce the AoI in wireless networks. One is that the use of NOMA provides   users more opportunities  to transmit, which means that the users can update their base station more frequently.  The other is that the use of NOMA can reduce access delay, i.e.,   the users are scheduled  to transmit earlier  than in the OMA case, which is useful to improve the freshness of the data available in the wireless network. 
\end{abstract}\vspace{-1em}


\section{Introduction}
In order to support the   services envisioned for  the sixth generation (6G) mobile network, such as  ultra massive machine type  communications (umMTC) and enhanced ultra-reliable low latency communications (euRLLC), it is critical to ensure the freshness of the data collected in the network \cite{8000687,6195689, 8930830,8123937,9380899 }. For example, as an important application of umMTC, smart cities require  the   data for  air quality control, traffic management, and critical infrastructure monitoring  to be timely and frequently    collected \cite{you6g,8766143}. We note that the conventional performance evaluation metrics, such as the ergodic data rate and bit error probability, are not adequate for measuring  the freshness of the data available in the network, which motivates a recently developed metric, termed the age of information (AoI). In particular, the AoI is defined as the time elapsed between the generation time and the received time of a successfully delivered  update, and the AoI achievable for   single-user transmission    has been rigorously characterized   in \cite{8000687,6195689}.  For energy constrained wireless networks, such as sensor networks, the use of energy harvesting is important, and the impact of energy harvesting on the AoI    has been studied in \cite{8606155}. For the scenario with correlated information  from multiple devices, a new metric, termed correlation-aware AoI, has been developed and optimized for unmanned aerial vehicle (UAV) networks in \cite{9606181}.  
Recently, the application of advanced physical layer communication techniques, such as hybrid automatic repeat request (H-ARQ) and cooperative communications, to improve the AoI of   wireless networks with one source-destination pair has been considered in \cite{9399662} and \cite{9615376}, respectively.



For the scenario with multiple users, the AoI analysis  is more challenging than that for the single-user scenario. This is due to the fact that multiple users are competing in the same transmission  medium, i.e., one user's update might be preempted by another's    and hence its AoI is affected by    the other users' transmission strategies. The AoI realized by  wireless transmission with various random access protocols has been characterized in \cite{8469047, 9695972,9388907}. We note that 
in many wireless networks, the potential collision between multiple users is  avoided by applying orthogonal multiple access (OMA) techniques, such as  frequency division multiple access (FDMA) and  time division multiple access (TDMA). In \cite{8995639}, the impact of these OMA techniques on the AoI has been studied, where   TDMA was shown to outperform  FDMA in terms of the averaged AoI. As  non-orthogonal multiple access (NOMA) is   more spectrally efficient than OMA, it is natural to consider  the use of NOMA for improving the AoI of wireless networks \cite{9693417}.  The authors of \cite{8845254} focused on   a two-user scenario and showed that  the spectral efficiency gain of NOMA over OMA indeed can be transferred to a reduction of the AoI.  In order to minimize the AoI, a dynamic policy to switch between NOMA and OMA was developed by formulating the AoI minimization problem as a Markov decision process  problem.   In \cite{9548950}, the impact of stochastic arrivals on the AoI achieved by NOMA was studied, where the performance gain of NOMA over OMA was shown to be significant for large arrival rates.  In \cite{9508961}, the AoI of NOMA assisted grant-free transmission was minimized by applying the tool of evolutionary game.  In \cite{9840754}, the AoI realized by reconfigurable intelligent surface (RIS) assisted NOMA transmission was analyzed,  and the application of NOMA to reduce the AoI of satellite communications was considered in~\cite{9734735}. 

This paper considers a general multi-user uplink communication network, where  OMA has already  been deployed  to serve the multiple users, i.e., there is a legacy network based on OMA. Because  TDMA outperforms FDMA in terms of   AoI, TDMA is  considered as an example for OMA in this paper.   Unlike the schemes reported in  \cite{9130084} and \cite{9771565}, which require a change of  the time frame structure of the legacy TDMA network, in this paper, NOMA is applied as an add-on to the TDMA legacy network. This   means that the AoI  is reduced with   minimal disruption to the legacy network, and hence the proposed NOMA protocols can be practically  implemented in various existing communication systems which are based on OMA.  The main contributions of the paper are the characterization
 of  the AoI achieved by the proposed NOMA protocols and  also the identification of   the benefits of using   NOMA to reduce the AoI of wireless networks, as explained in the following:
\begin{itemize}
\item For the case where each user's update is generated at the beginning of its transmit time slot, cognitive-radio inspired NOMA (CR-NOMA) is applied to ensure that each user has two opportunities  to deliver its update to the base station in each TDMA time frame \cite{Zhiguo_CRconoma}.    A closed-form expression for the AoI realized by CR-NOMA is developed to facilitate the  performance evaluation, and  an asymptotic analysis reveals  that at high signal-to-noise ratio (SNR), CR-NOMA and TDMA realize the same AoI. This conclusion is expected since, at high SNR, each user   needs one transmission only to ensure that its update is successfully delivered to the base station. However, at low SNR, the use of CR-NOMA can result in significant AoI reduction compared to TDMA, which demonstrates  that one benefit of using NOMA  is to offer users more chances to transmit, i.e., by using NOMA the users can update their base station more frequently.

\item For the case where each user's update is generated at the beginning of a TDMA time frame, a modified CR-NOMA protocol is developed to demonstrate another benefit of using NOMA   to improve the AoI. In particular, the modified CR-NOMA protocol effectively reduces   the users' access delay,  i.e.,   the users are scheduled  to transmit earlier  than in the TDMA case, which is useful for improving the freshness of the data available at the base station.  For example,     a user which is scheduled in a time slot close to the end of a TDMA frame  experiences severe access delay and hence suffers from a large AoI with TDMA, because it has to wait for a long time before it can send its update which has been generated at the beginning of the frame. The use of NOMA   ensures that this user can jump the queue and transmit earlier than with TDMA. An exact expression for the AoI achieved with the modified CR-NOMA protocol is obtained, and an asymptotic analysis demonstrates  that the use of NOMA yields a reduction of AoI at high    SNR, compared to TDMA. 
\end{itemize} 

   The remainder  of this paper is organized as follows. In Section \ref{section 2}, the system model and the considered data generation models are described.   In Sections \ref{section 3} and \ref{section 4}, two CR-NOMA based transmission protocols are introduced  and their impact on the AoI   is characterized.   Simulation results are presented in Section \ref{section 4}, and the paper is concluded in Section \ref{section 5}. Finally,   all the proofs are collected in the appendix.

\section{System Model}\label{section 2}
Consider a  communication network with $M$ users, denoted by ${\rm U}_m$, $1 \leq m\leq M$, sending their updates to the same base station. Assume that  an OMA based legacy network  has been employed    to serve these users. Because TDMA  yields smaller AoI than FDMA \cite{8995639},   TDMA is used as an example of OMA in this paper. There are $M$ time slots in each TDMA time frame and ${\rm U}_m$ is scheduled to transmit in the $m$-th time slot of each frame with   transmit power $P$, as shown in Fig. \ref{fig1}. Denote  the duration of each time slot   by $T$ s, and    the start of the $m$-th time slot in the $i$-th frame   by $t_i^m$, $1\leq m\leq M$. Therefore, with TDMA, each user can deliver one update to the base station every $MT$ s. 
    \begin{figure}[t]\centering \vspace{-1em}
    \epsfig{file=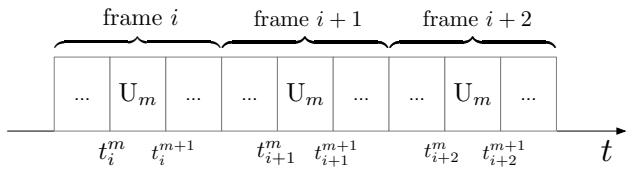, width=0.5\textwidth, clip=}\vspace{-0.5em}
\caption{Considered TDMA time frame structure.   \vspace{-1em}    }\label{fig1}   \vspace{-1em} 
\end{figure}

\subsection{AoI Model}
AoI indicates the freshness of the updates successfully delivered to  the base station and can be defined as follows. At time $t$, denote by $T_m(t)$  the
generation time of the freshest update   successfully delivered from  ${\rm U}_m$ to 
the base station.  Therefore,  ${\rm U}_m$'s instantaneous AoI   is defined as follows~\cite{8000687}: 
\begin{align}
\Delta_m(t) = t - T_m(t).
\end{align}
The normalized overall average AoI of the considered network is given by
\begin{align}
\bar{\Delta} =\frac{1}{M} \sum^{M}_{m=1} \underset{\bar{\Delta}_m}{\underbrace{\underset{T\rightarrow\infty}{\lim} \frac{1}{T}\int^{T}_{0}\Delta_m(t) dt}}.
\end{align} 

 \subsection{Data Generation Models}
When data is generated has significant  impact on the   AoI, and  the following two types of data generation are considered in this paper:

\subsubsection{The generate-at-will (GAW) model}  This model assumes that a new update is generated right  before the transmit time slot of the user \cite{8000687, 8845254,9130084 }. For example, if a user decides to transmit in the $m$-th time slot of the $i$-th frame, the update to be sent in this time slot is generated at $t_i^m$, as shown in Fig. \ref{fig2cc}(a).  As a widely used data generation model,  GAW has the advantage of improving  the freshness of the data by always transmitting  a freshly generated update. However, GAW    potentially  results in   high system complexity as  a user is required to repeatedly generate updates if the user is offered multiple chances to transmit in a short period.  

\subsubsection{The generate-at-request (GAR) model} With this model,  the base station requests each user to generate an update  at the beginning  of each time frame, instead of each time slot, as shown in Fig. \ref{fig2cc}(b). If retransmission is carried out within one frame, the same update will be sent. GAR is crucial  to synchronized sensing, and hence   important  in many Internet of Things (IoT) applications, such as  structural health monitoring and autonomous driving \cite{8000687x,9796958}.  GAR  can lead to larger AoI than GAW, since a user's access delay, i.e., the duration between   the   generation time of an update and the corresponding  transmit time, is included in the calculation of AoI. However, compared to GAW, GAR  can reduce system complexity and    energy consumption, since  GAR avoids asking the users to repeatedly generate updates for retransmission.  
 
    \begin{figure}[t]\centering \vspace{-1em}
    \epsfig{file=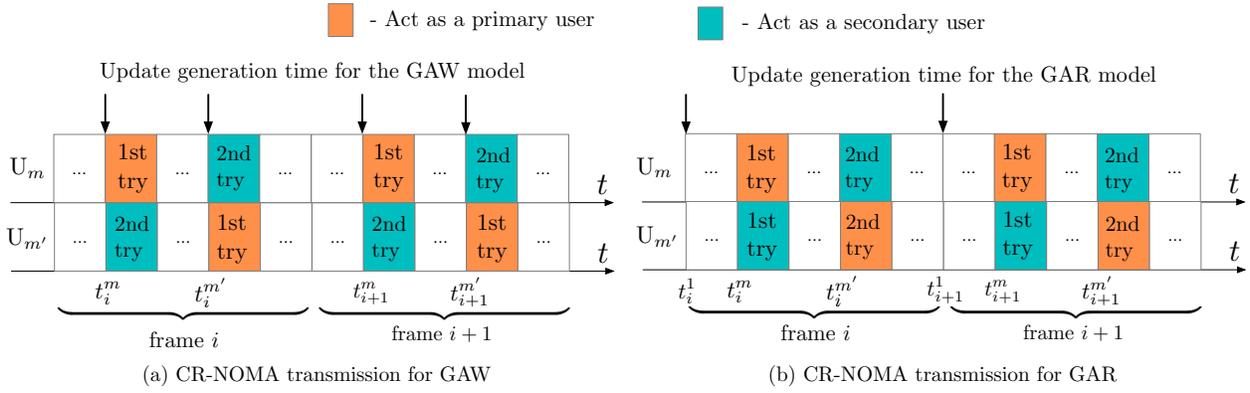, width=1\textwidth, clip=}\vspace{-0.5em}
\caption{Illustration for the two proposed CR-NOMA transmission schemes.   \vspace{-1em}    }\label{fig2cc}   \vspace{-1em} 
\end{figure}

\subsection{CR-NOMA Transmission}\label{subsection crnoma}
 CR-NOMA can be used as an add-on to TDMA to improve the freshness of the data collected  in the   network, as explained in the following.   To reduce   system complexity, the following simple user pairing scheme is considered. In particular, ${\rm U}_m$ and ${\rm U}_{m'}$, where $1\leq m \leq \frac{M}{2}$ and $m'=m+\frac{M}{2}$, are paired together to share the spectrum, and are allowed to transmit simultaneously in the     $m$-th and the $m'$-th time slots of each frame \footnote{For illustrative purposes, it is assumed that $M$ is an even number, but the proposed protocol can be straightforwardly extended to  the case with an odd number of users, e.g.,  one of the users can be simply served by  TDMA.  }.

The aim of       CR-NOMA transmission    is to increase the likelihood for each user to deliver one update to the base station every $MT$ s, compared to TDMA. Depending on the used data generation model,  the application of CR-NOMA transmission   is slightly different. In brief, for GAW, ${\rm U}_m$ and  ${\rm U}_{m'}$ use  the $m$-th and the $m'$-th time slots of each frame for their first tries for updating, respectively, as shown in Fig. \ref{fig2cc}(a). If their first tries are not successful,  ${\rm U}_m$ uses the $m'$-th time slot of the current frame, and ${\rm U}_m$ uses the $m$-th time slot of the next frame for   retransmission. For GAR, both the users use the $m$-th time slot for their first tries and  the $m'$-th time slot for their second tries, as  shown in Fig. \ref{fig2cc}(b).  The details of CR-NOMA transmission are provided  as follows.  

\subsubsection{CR-NOMA  Transmission for GAW}   In the $m$-th time slot of the $i$-th frame, ${\rm U}_m$ is treated as the primary user and is scheduled to transmit with transmit power $P$, in the same manner as   TDMA. If this transmission is not successful, ${\rm U}_m$ is offered another chance to transmit in the $m'$-th time slot as the secondary user with transmit power $P^S$, where     ${\rm U}_{m'}$ is     treated as the primary user in this time slot.    To ensure that the implementation of NOMA is transparent to the primary user, the secondary user's signal is decoded in the first stage of successive interference cancellation (SIC) at the base station and its data rate needs to be capped \cite{Zhiguo_CRconoma}. For example, in the $m'$-th time slot,  ${\rm U}_m$ is the secondary user and its data rate, denoted by $R_{m,i}^S$, is capped as follows:  
\begin{align}\label{crnoma}
R_{m,i}^S \leq  \log\left( 1+\frac{P^S|h_{m}^{i,m'}|^2}{P|h_{m'}^{i,m'}|^2+1}\right) ,
\end{align} 
where the binary logarithm is used, and  ${\rm U}_{n}$'s channel gain in the $k$-th time slot of the $i$-th frame is denoted by $h_n^{i,k}$. We note that the noise power is assumed to be normalized, and hence $P$ and $P^S$ are the effective transmit SNRs. The users' channel gains in different time slots are assumed to be independent and identically distributed (i.i.d.) and follow the complex Gaussian distribution with zero mean and unit variance. 

Similarly, in the $m'$-th time slot of the $i$-th frame, ${\rm U}_{m'}$ is scheduled to transmit as the primary user,  in the same manner  as in   TDMA. If this transmission is not successful, the user will   use the $m$-th time slot of the next frame to transmit  a new update. Because ${\rm U}_{m'}$ is the secondary user in the $m$-th time slot,  its data rate in this time slot needs to be   capped similarly as in \eqref{crnoma}.  
  
\subsubsection{CR-NOMA  Transmission for GAR}    Recall that with GAR, the users' updates are generated at the beginning of each TDMA frame. On the one hand,  in order to reduce the access delay,  ${\rm U}_{m'}$ will not wait for the $m'$-th time slot, but  use the $m$-th time slot to deliver its update. We note that  ${\rm U}_m$ is treated as the primary user in the $m$-th time slot, which means that ${\rm U}_{m'}$'s data rate needs to be  capped similarly as in \eqref{crnoma}.  Only if  ${\rm U}_{m'}$'s transmission in the $m$-th time slot is not successful, ${\rm U}_{m'}$ carries out a retransmission  in the $m'$-th time slot of the current frame. 

On the other hand, ${\rm U}_{m}$'s transmission strategy for GAR in the $m$-th time slot is the same as that for GAW. If ${\rm U}_{m}$ requires a retransmission, it will be   treated as the secondary user in the $m'$-th time slot, and its achievable   data rate in this time slot depends  on   its partner's transmission strategy, which is different from that for GAW. In particular, if ${\rm U}_{m'}$'s transmission in the $m$-th time slot is not successful, a retransmission from ${\rm U}_{m'}$ in the $m'$-th time slot is needed, and hence ${\rm U}_{m}$'s   data rate in this time slot is capped as in \eqref{crnoma}. However, if ${\rm U}_{m'}$'s transmission in the $m$-th time slot is   successful, ${\rm U}_{m'}$ keeps silent in the $m'$-th time slot, which means that ${\rm U}_{m}$'s achievable data rate in this time slot is simply given by: 
\begin{align}\label{crnoma2}
R_{m,i}^S =  \log\left( 1+ P^S|h_{m}^{i,m'}|^2 \right) . 
\end{align} 

The AoI achieved by  CR-NOMA transmission for the two data generation models will be analyzed in the following two sections, respectively. 

\section{AoI of CR-NOMA Transmission for GAW} \label{section 3}
In this section, the AoI achieved by CR-NOMA is studied     the GAW model, where   TDMA is used as a benchmarking scheme.  The benefit of using CR-NOMA to reduce the AoI can be clearly illustrated with the example shown in Fig. \ref{fig2}. Because GAW is used, the instantaneous AoI is reduced to $T$ whenever an update is successfully delivered to the base station. For the illustrated example,  ${\rm U}_m$ fails to deliver its update to the base station in the $m$-th time slot of  the $(i+1)$ frame.  With TDMA, the user has to wait until the next TDMA frame. However, with NOMA, the user has another chance for retransmission in the current frame, which is helpful to reduce the AoI.

 \begin{figure}[t] \vspace{-0em}
\begin{center}
\subfigure[AoI evolution of TDMA Transmission]{\label{fig2a}\includegraphics[width=0.45\textwidth]{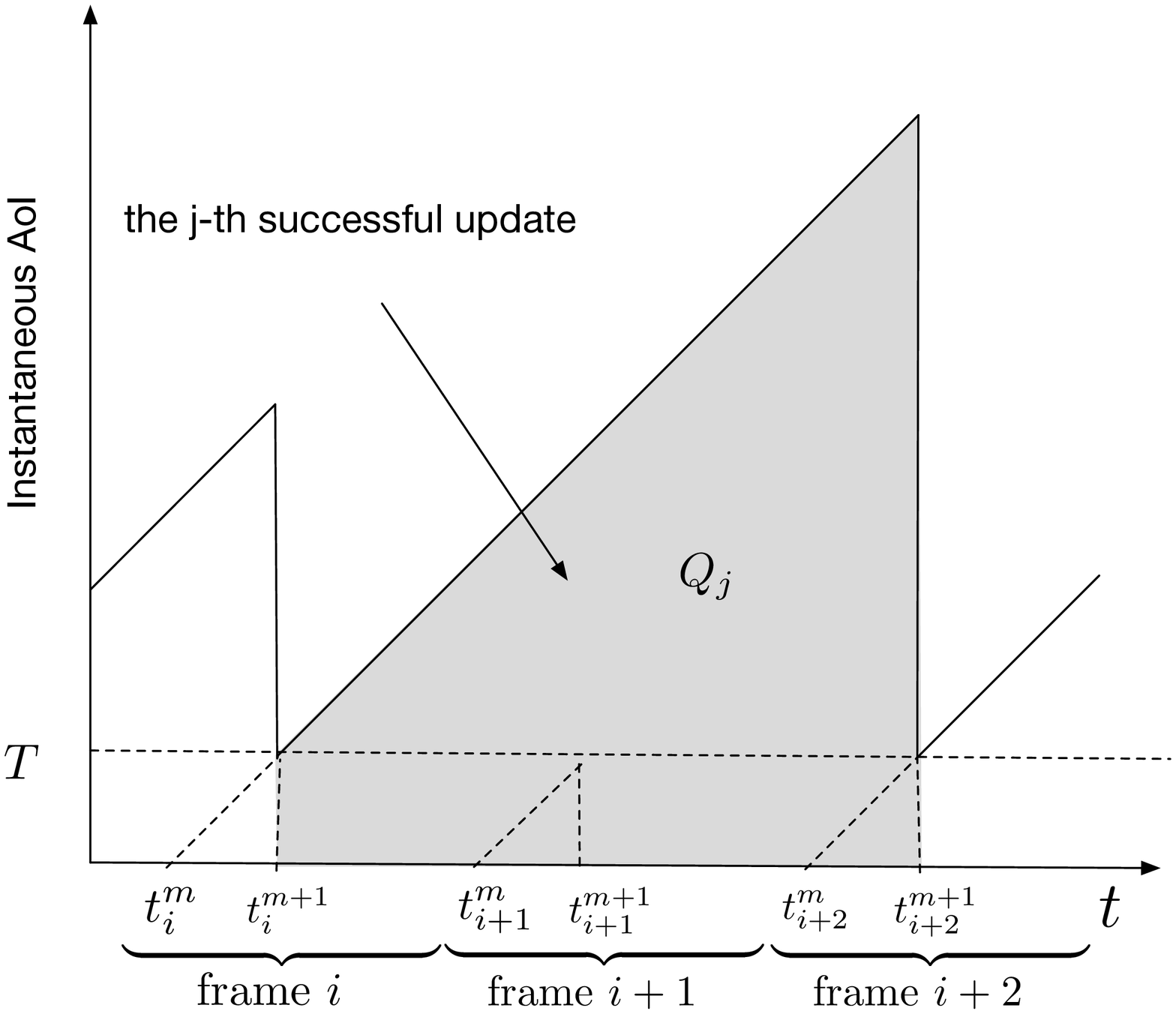}}\hspace{2em}
\subfigure[AoI evolution of NOMA Transmission]{\label{fig2b}\includegraphics[width=0.45\textwidth]{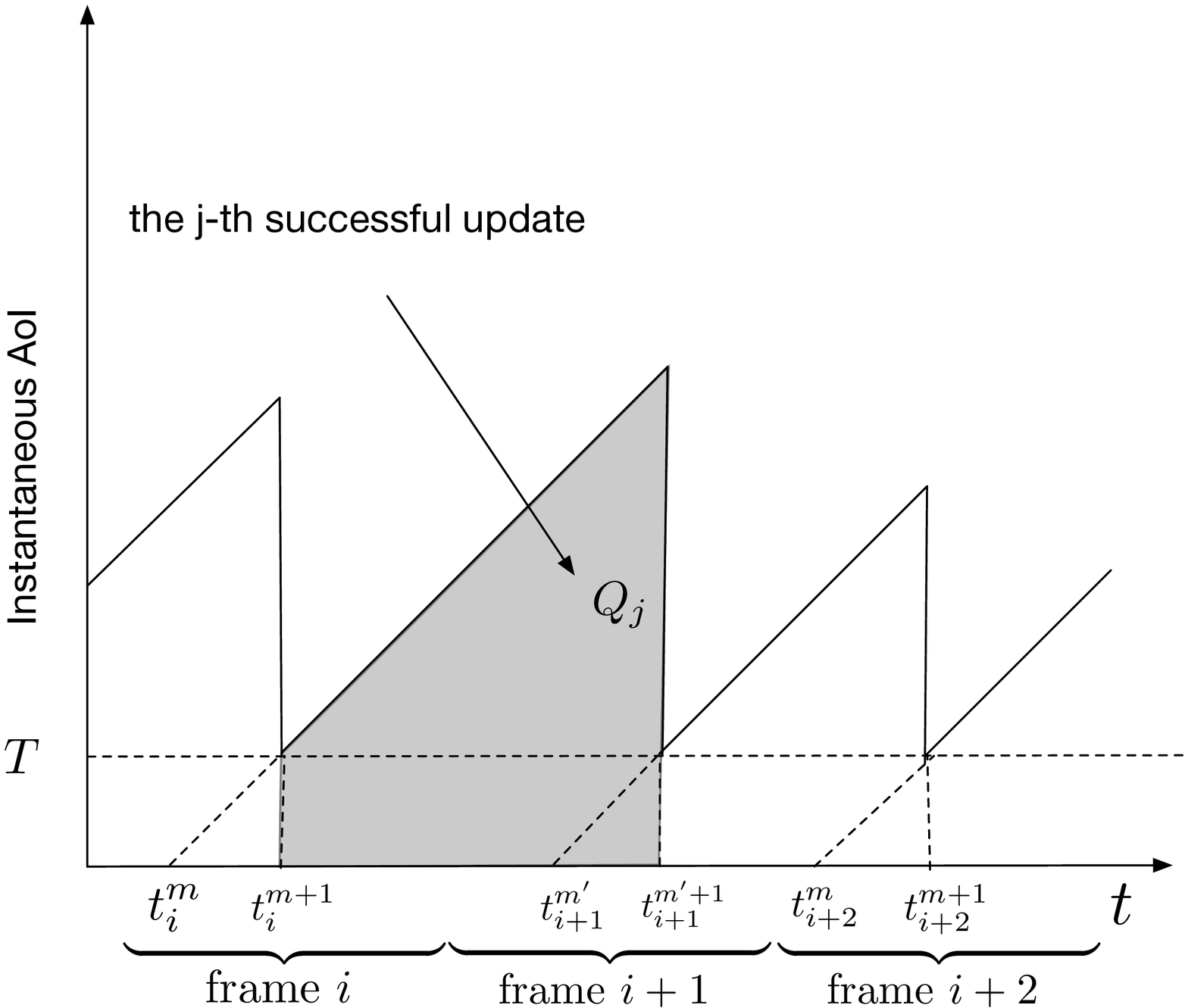}} \vspace{-1em}
\end{center}
\caption{ Illustration of AoI evolution for the two considered  transmission protocols.  The generate-at-will model    is considered. For TDMA, ${\rm U}_m$ relies on the $m$-th time slot of each frame, i.e., from $t_i^m$ to $t_i^{m+1}$ in frame $i$. For the illustrated example, ${\rm U}_m$ successfully delivers its update in frames $i$ and $(i+2)$, but fails in frame $(i+1)$.  For NOMA, the user also fails in the $m$-th time slot of frame $(i+1)$, but successfully delivers its update in the $m'$-th $\left(m'=m+\frac{M}{2}\right)$ time slot of the same frame, which improves data freshness.    \vspace{-1em} }\label{fig2}\vspace{-1em}
\end{figure}

\subsection{AoI Realized for TDMA}\label{sub tdma} The AoI realized for TDMA can be straightforwardly analyzed as follows. Without loss generality, we focus on ${\rm U}_m$'s average AoI realized for TDMA, denoted by $\bar{\Delta}^T_m$. Further denote the number of frames between the     $(j-1)$-th and the $j$-th successful updates by $x_j$.  In Fig. \ref{fig2}(a), an example with $x_j=2$ is shown. Therefore, finding the average AoI for TDMA is equivalent to find the area of the shaded region in Fig. \ref{fig2}(a), denoted by $Q_j$, which means that $\bar{\Delta}_m^T$ can be expressed as follows:
\begin{align}\label{tdma_aoi}
\bar{\Delta}_m^T = \underset{J\rightarrow \infty}{\lim} \frac{\sum^{J}_{j=1} Q_j }{\sum^{J}_{j=1}x_jMT} =T+ \frac{MT}{2} \frac{  \mathcal{E}\{X^2\}}{ \mathcal{E}\{X\}},
\end{align}
where $J$ denotes the total number of the successful updates, $Q_j=x_jMT^2+\frac{1}{2}x_j^2M^2T^2$, $\mathcal{E}\{X\}=\underset{J\rightarrow \infty}{\lim}\frac{1}{J}\sum^{J}_{j=1}x_j$, and $\mathcal{E}\{X^2\}=\underset{J\rightarrow \infty}{\lim}\frac{1}{J}\sum^{J}_{j=1}x_j^2$. Since the users' channel gains are assumed to be i.i.d., $x_j$ follows the geometric distribution, i.e.,    the probability mass function
 of $x_j$ is given by $\mathbb{P}(X=x_j)=p_e^{x_j-1} (1-p_e) $, where $p_e$ denotes the probability for the event that a user fails to deliver an update   in a time slot with $T$ s. By using the assumptions  that each update contains $N$ bits and the users' channel gains are complex Gaussian distributed,    $p_e=\mathbb{P}(T\log (1+P|h_{m}^{i,m'}|^2)\leq N)=1-e^{-\frac{\epsilon}{P}}$, where $\epsilon=2^{\frac{N}{T}}-1$. As a result, the normalized average AoI achieved by TDMA can be obtained as follows:
 \begin{align}\label{tdma}
 \bar{\Delta}^T \overset{(1)}{=}\bar{\Delta}_m^T  \overset{(2)}{=}  T+ \frac{MT}{2} \left(2e^{\frac{\epsilon}{P}}-1 \right),
 \end{align}
 where the first step follows by the fact that the users experience the same AoI in TDMA, and the second step   follows by   using the mean and the variance of the geometric distribution.

\subsection{AoI Realized for CR-NOMA} 
The analysis of the AoI realized with CR-NOMA is more challenging than that for TDMA. The reason is that each user has two transmission opportunities  in each TDMA frame, which means that the time interval  between the two adjacent successful updates is not always a multiple   of $MT$.  The following lemma provides a closed-form expression for the AoI realized by CR-NOMA.

 \begin{lemma}\label{lemma1}
For the case of the GAW model, the normalized overall average AoI realized by   CR-NOMA   is given by 
 \begin{align}
\bar{\Delta}^N  = T+ \Delta(p_0, p_m, p_{m'}),
\end{align}
where $\Delta(x, y, z)$ is defined as follows:
\begin{align}
\Delta(x, y, z) = \frac{MT}{4} \frac{ 2(y+z)^2 (1+x) +   yz(1-x)^2}{ (y+z)^2  (1-x) },
\end{align}
 $p_m = e^{-\frac{\epsilon}{P^S}}$,   $
p_{m'} =    \left(1-e^{- \frac{\epsilon}{P^S}}\right)   e^{-\frac{\epsilon}{P^S} }\frac{1}{1+\frac{P\epsilon}{P^S}} $, 
and $
p_0 
= \left(1-e^{- \frac{\epsilon}{P^S}}\right)  \left(1- \frac{e^{-\frac{\epsilon}{P^S} }}{1+\frac{P\epsilon}{P^S}}\right)   $.
\end{lemma}
\begin{proof}
See Appendix \ref{proof1}.
\end{proof}

Based on the closed-form expression of the AoI given in Lemma \ref{lemma1}, an asymptotic analysis  can be carried out to obtain an insightful understanding of the impact of NOMA on the AoI.  For example,  consider the following high SNR scenario,  $P^S=P\rightarrow \infty$, which means that $p_m\approx 1-\frac{\epsilon}{P}$. $p_{m'}$ can be approximated as follows:
 \begin{align} 
p_{m'} 
=&  \left(1-e^{- \frac{\epsilon}{P^S}}\right)   e^{-\frac{\epsilon}{P^S} }\frac{1}{1+\frac{P\epsilon}{P^S}}  \approx  \frac{ \epsilon}{P(1+\epsilon)}  .
\end{align}
Furthermore, $p_0 $ can be approximated at high SNR as follows:
\begin{align} 
p_0 \nonumber
=&  \left(1-e^{- \frac{\epsilon}{P^S}}\right)  \left(1- \frac{e^{-\frac{\epsilon}{P^S} }}{1+\frac{P\epsilon}{P^S}}\right)  
\approx    \frac{\epsilon^2}{P (1+\epsilon)}.
\end{align}
By using the high-SNR approximations of $p_0$, $p_m$ and $p_{m'}$, the AoI achieved by CR-NOMA can be approximated as follows:
 \begin{align}\label{noma1}
\bar{\Delta}^N  &\approx T+ \frac{MT}{4} \frac{ 2(p_m +p_{m'})^2   +   p_mp_{m'} }{ (p_m+p_{m'})^2   }\\\nonumber
&\overset{(1)}{\approx} T+ \frac{MT}{4} \frac{ 2(p_m )^2     }{ (p_m )^2   } = T+\frac{MT}{2},
\end{align}
 where step 1 follows by the fact that $p_m\gg p_{m'}$ at high SNR. 
 
 On the other hand,   the AoI realized by TDMA can be approximated at high SNR as follows:
  \begin{align}\label{tdma1}
 \bar{\Delta}^T =  T+ \frac{MT}{2} (2e^{\frac{\epsilon}{P}}-1 ) \approx T+ \frac{MT}{2}.
 \end{align}
 Comparing the AoI shown in \eqref{noma1} and \eqref{tdma1}, the following corollary can be obtained.  
 
\begin{corollary}\label{corollary1}
For the case of the GAW model, at high SNR, i.e.,   $P=P^S\rightarrow \infty$,   the normalized average AoI achieved by CR-NOMA is   same as that of TDMA . 
\end{corollary}

{\it Remark 1:} The conclusion shown in Corollary  \ref{corollary1} is expected since, at high SNR, ${\rm U}_m$'s first try, i.e., its transmission in the $m$-th time slot, is almost guaranteed to be successful. Therefore, retransmission is not needed at high SNR, and CR-NOMA is reduced  to TDMA, which explains why the two protocols achieve  the same AoI at high SNR. However, it is important to point out that the use of CR-NOMA can result in a significant performance gain over TDMA in the low  SNR regime, as shown in the simulation section. 

{\it Remark 2:} The use of CR-NOMA means that a user may need to transmit  two times every $MT$ s, which can cause  a higher  energy consumption than for  TDMA. In order to avoid this drawback,  prior to its transmit time slot,  a user can first calculate its data rate supported by CR-NOMA, and transmit only    if this data rate is   sufficient to deliver its update. This implementation requires that each user knows its own channel state information (CSI) as well as that of its partner. This CSI assumption can be realized as follows. The base station   first broadcasts a pilot signal before a  time slot, and then each user   carries  out channel estimation individually. In addition, the CSI of each user's partner can be obtain either via device-to-device communication between the partners or via a dedicated control channel.

 \begin{figure}[t] \vspace{-0em}
\begin{center}
\subfigure[AoI evolution of TDMA Transmission]{\label{fig3a}\includegraphics[width=0.45\textwidth]{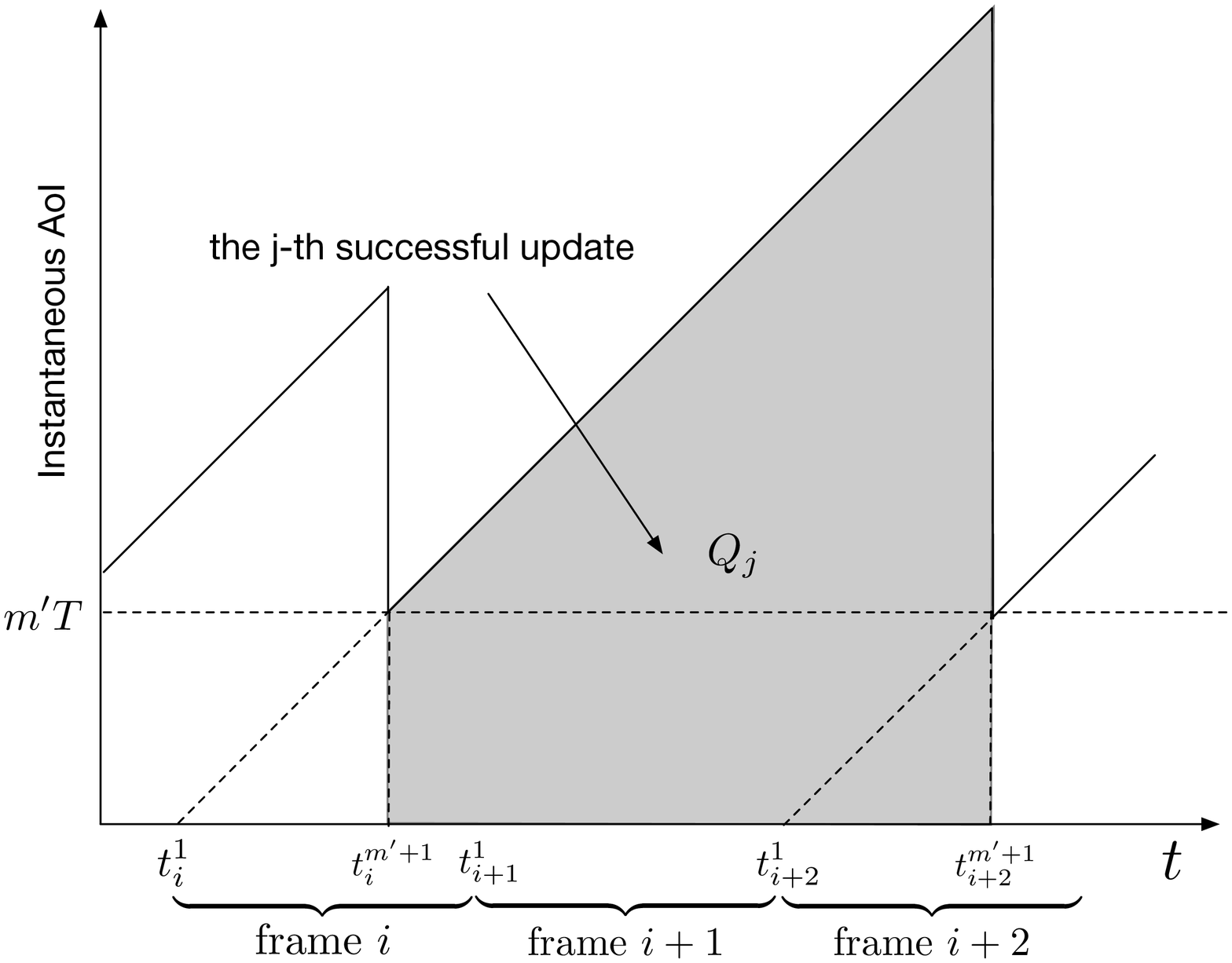}}\hspace{2em}
\subfigure[AoI evolution of NOMA Transmission]{\label{fig3b}\includegraphics[width=0.45\textwidth]{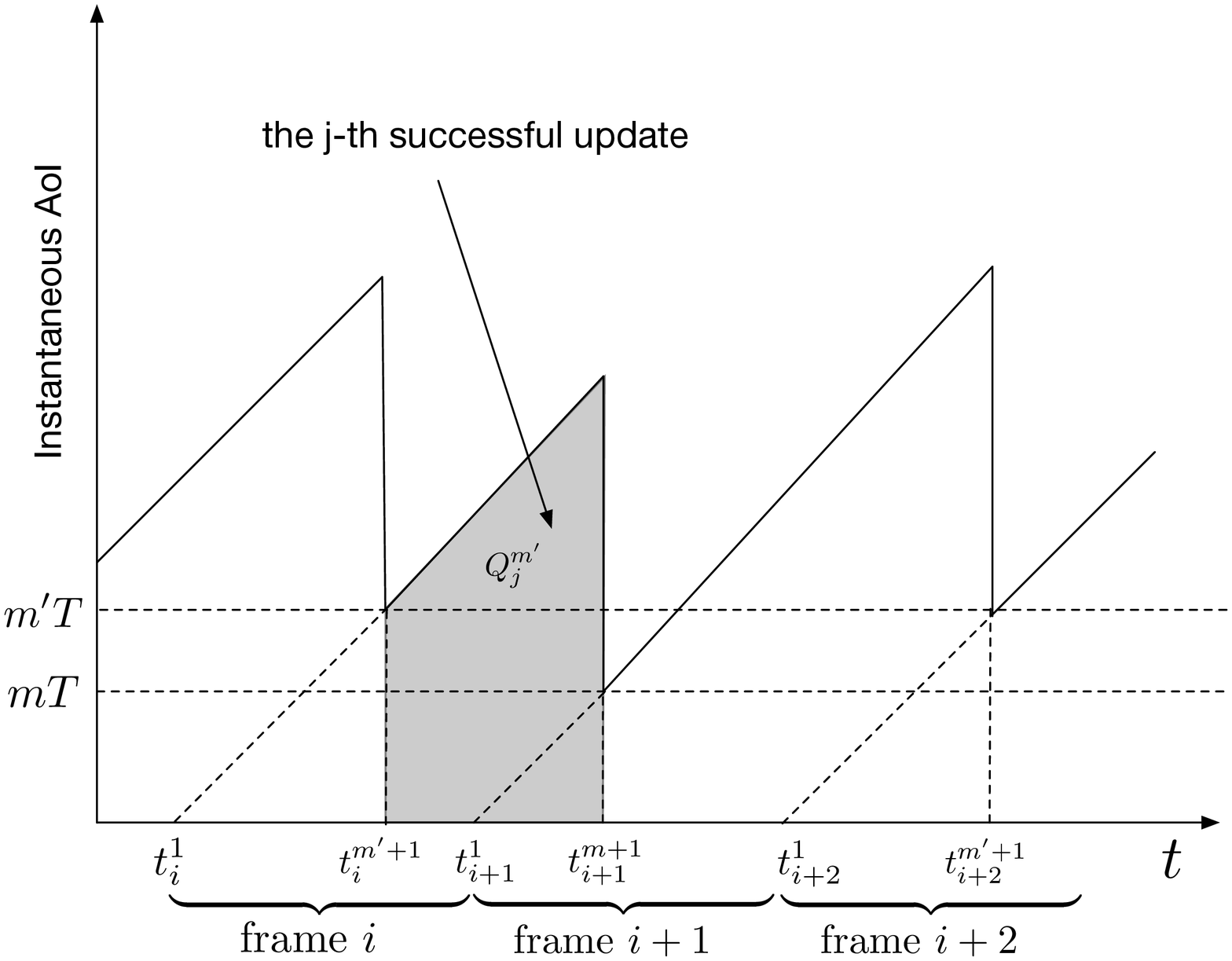}} \vspace{-1em}
\end{center}
\caption{Illustration of AoI evolution for the   considered    protocols.  The generate-at-request model    is considered.  For the illustrated example, ${\rm U}_{m'}$ successfully delivers its update in frames $i$ and $(i+2)$, but fails in the $m'$-th time slot of frame $(i+1)$ for TDMA. For NOMA, the user    successfully delivers its update by using the $m$-th  time slot of frame $(i+1)$, which improves data freshness.   \vspace{-1em} }\label{fig3}\vspace{-2em}
\end{figure}

\section{AoI of CR-NOMA Transmission for GAR} \label{section 4}
Unlike the GAW model, the GAR model requires the users to generate their updates at the beginning of each frame, which means that the access  delay needs to be included in the AoI.  The impact of the access delay on the AoI can be illustrated by using   Fig. \ref{fig3}, where   ${\rm U}_{m'}$'s AoI experience  is considered. With TDMA, ${\rm U}_{m'}$ has to wait until  the $m'$-th time slot of each frame for its transmission. This   means that   the user's instantaneous AoI drops   to    $m'T$,  since its update is generated at the beginning of the frame.   With CR-NOMA, ${\rm U}_{m'}$ can  also use the $m$-th time slot of each frame, which results in two benefits for reducing the AoI. One is that with two chances to   transmit every $MT$ s, the likelihood of a failed update is reduced, which is similar to the      GAR case. The second benefit  is that the use of NOMA can effectively reduce the access delay, since  ${\rm U}_{m'}$ can transmit earlier than for   TDMA  by using the $m$-th time slot instead of the $m'$-th time slot, and   its  instantaneous AoI can  drop  to $mT$.

\subsection{AoI Realized for TDMA} To be consistent with the illustration of Fig. \ref{fig3}, ${\rm U}_{m'}$'s AoI is considered, but ${\rm U}_{m}$'s AoI can be analyzed  similarly. As shown in Fig. \ref{fig3a},  in TDMA,   ${\rm U}_{m'}$ is allowed to use the $m'$-th time slot only, i.e.,  its transmit time slot  ends  at $t_i^{{m'}+1}$ in the $i$-th frame,  where it is important to point out that its update  is generated at $t_i^1$.  Again denote  by $x_j$ the number of frames between the $(j-1)$-th and the $j$-th successful updates, where an example with $x_j=2$ is shown in Fig. \ref{fig3}.

Similar to the GAW case, the time span between the two adjacent successful updates is   $x_jMT$; however, the AoI does not drop  to $T$ whenever there is a successful update. Instead, the AoI drops  to ${m'}T$ only, because ${\rm U}_{m'}$'s update is generated at the beginning of the frame.  Therefore, ${\rm U}_{m'}$'s  averaged AoI achieved by TDMA is given by
\begin{align}
\bar{\Delta}_{m'}^T =& \underset{J\rightarrow \infty}{\lim} \frac{\sum^{J}_{j=1} Q_j }{\sum^{J}_{j=1}x_jMT} \\\nonumber
=&  \underset{J\rightarrow \infty}{\lim} \frac{\sum^{J}_{j=1} \left(
{m'}Tx_jMT + \frac{1}{2}x_j^2M^2T^2
\right) }{\sum^{J}_{j=1}x_jMT}
= {m'}T+ \frac{MT}{2}   \frac{\mathcal{E}\{  x_j^2\}
  }{\mathcal{E}\{x_j\} }.
\end{align}
By using   steps similar to those in Section \ref{sub tdma}, ${\rm U}_{m'}$'s  averaged AoI achieved by TDMA for the GAR model can be obtained as follows:
 \begin{align}\label{tdma_aoi new}
 \bar{\Delta}_{m'}^T  =  {m'}T+ \frac{MT}{2}\left (2e^{\frac{\epsilon}{P}}-1 \right).
 \end{align}

Comparing \eqref{tdma_aoi} to \eqref{tdma_aoi new}, one can find that with the GAR model, the users experience a larger AoI, which is mainly  due to the access delay, i.e., a user generates its update     at the beginning of one frame but has to wait for its turn to deliver  the update.  
%
%
 
\subsection{AoI Realized for CR-NOMA}\label{subsection IV.B}
 The AoI realized by CR-NOMA in the GAR case is   challenging to analyze due to the fact that  a user's instantaneous AoI does not always drop  to the same value, which is different from the GAW case. In particular, in the GAW case, a user's instantaneous AoI is always reduced to $T$, regardless which time slot is used for update delivery. However, in the GAR case,   if ${\rm U}_m$ successfully delivers its update in the $m$-th time slot of a frame, its instantaneous AoI   drops to $mT$, whereas its instantaneous AoI   drops  to $m'T$ if the transmission in the $m'$-th time slot is successful. What makes the AoI analysis more complicated is that the calculation of the shaded region $Q_j^{m'}$ shown in Fig. \ref{fig3}  does not only depend on which time slot is used for the $(j-1)$-th successful update but also on which time slot is used for the $j$-th successful update.  Furthermore, whether a user can successfully deliver its update to the base station is  also affected by its partner's transmission strategy, as discussed in Section \ref{subsection crnoma}.   The following lemma provides a closed-form expression for the AoI realized by CR-NOMA in the GAR case. 
\begin{lemma}\label{lemma2}
 For the case of the GAR model, ${\rm U}_{k}$'s   average AoI realized by   CR-NOMA   is given by  
\begin{align}\label{lemma2main}
\bar{\Delta}^N_k =& \Delta_{k,0}  + \Delta(p_{0k}, p_{mk}, p_{m'k}),
\end{align}
where $k\in\{m, m'\}$, $1\leq m \leq \frac{M}{2}$,   $\Delta_{k,0}  $ is given by
\begin{align}
\Delta_{k,0}  =& \frac{(1-p_{0k})^2}{(p_{mk}+p_{m'k})^2  }   \left[ \left(  (p_{mk} +p_{m'k} )  \frac{mTp_{mk}}{(1- p_{0k})^2} + \frac{p_{m'k}}{2} \frac{mTp_{mk}}{1- p_{0k}}\right)\right. \\  &+ \left.  \left( (p_{mk} +p_{m'k} )  \frac{m'Tp_{m'k}}{(1- p_{0k})^2} - \frac{p_{mk}}{2} \frac{m'Tp_{m'k}}{1- p_{0k}}\right)  \right],
\end{align}
 $p_{mm} = e^{-\frac{\epsilon}{P^S}}$,    $
p_{m'm} =     \left(1-e^{-\frac{\epsilon}{P}}-  e^{- \frac{\epsilon}{P^S}}\tau\right)   e^{-\frac{\epsilon}{P^S} }\frac{1}{1+\frac{P\epsilon}{P^S}} 
+ e^{- \frac{2\epsilon}{P^S} }\tau  $, $
p_{0m} 
= \left(1-e^{-\frac{\epsilon}{P}}-  e^{- \frac{\epsilon}{P^S}}\tau\right)  \left(1- \frac{e^{-\frac{\epsilon}{P^S} }}{1+\frac{P\epsilon}{P^S}}\right)
 + e^{- \frac{\epsilon}{P^S} }\tau\left(1-e^{- \frac{\epsilon}{P^S}}\right)  $ , 
$p_{0m'} =  \left(1-e^{-\frac{\epsilon}{P^S}}\frac{1}{1+\frac{\epsilon P}{P^S}}\right)(1-e^{-\frac{\epsilon}{P}})$, $p_{mm'}  =e^{-\frac{\epsilon}{P^S}}\frac{1}{1+\frac{\epsilon P}{P^S}}$,   $p_{m'm'} =  \left(1-e^{-\frac{\epsilon}{P^S}}\frac{1}{1+\frac{\epsilon P}{P^S}}\right)e^{-\frac{\epsilon}{P}}$ and $\tau=\frac{1-e^{-\left(\frac{\epsilon}{P^S}P+1\right)\frac{\epsilon}{P} }}{\frac{\epsilon}{P^S}P+1}$. The normalized overall average AoI realized by CR-NOMA is given by $\bar{\Delta}^N = \frac{1}{M}\sum^{\frac{M}{2}}_{m=1} \left( \bar{\Delta}^N_m+\bar{\Delta}^N_{m'}\right)$.  
\end{lemma}
\begin{proof}
See Appendix \ref{proof3}. 
\end{proof}

Although the AoI expression shown in Lemma \ref{lemma2} is lengthy and complicated, it can be used to obtain an insightful understanding of the impact of NOMA on the AoI. Let us consider the   high SNR scenario, i.e., $P^S=P\rightarrow \infty$. Because the use of NOMA reduces the access delay for ${\rm U}_{m'}$, it is expected that ${\rm U}_{m'}$'s AoI experience for TDMA and NOMA will  be significantly different, and hence we focus on ${\rm U}_{m'}$'s AoI   in the following.

With some straightforward algebraic manipulations, at high SNR,  the following approximations can be obtained: $p_{0m'} \approx  \frac{\epsilon}{1+ \epsilon }  \frac{\epsilon}{P}$, $p_{mm'}  \approx \frac{1-\frac{\epsilon}{P}}{1+ \epsilon }$, and $p_{m'm'} \approx   \frac{\epsilon}{1+ \epsilon } \left(1-\frac{\epsilon}{P}\right)$, which lead to the following approximation:    $p_{mm'} +p_{m'm'}\approx  1-\frac{\epsilon}{P}$. 
On the one hand, by using these  approximations, the first term in \eqref{lemma2main}, $\Delta_{m',0}  $, can be approximated at high SNR as follows:
\begin{align} 
\Delta_{m',0}  
 \approx& \frac{1}{(1-\frac{\epsilon}{P})^2  }  \left[ \left( \left (1-\frac{\epsilon}{P}\right)   mT \frac{1-\frac{\epsilon}{P}}{1+ \epsilon } +  \frac{\epsilon \left(1-\frac{\epsilon}{P}\right)}{2(1+ \epsilon) }   mT \frac{1-\frac{\epsilon}{P}}{1+ \epsilon } \right)\right. \\\nonumber &+ \left.  \left( \left(1-\frac{\epsilon}{P}\right)   m'T\frac{\epsilon \left(1-\frac{\epsilon}{P}\right)}{1+ \epsilon }  -  \frac{1-\frac{\epsilon}{P}}{2(1+ \epsilon )} m'T\frac{\epsilon \left(1-\frac{\epsilon}{P}\right)}{1+ \epsilon }\right)  \right]\\\nonumber
 =&  \frac{1}{(1+\epsilon)^2}  \left[   mT(1+\epsilon)   +  \frac{\epsilon }{2  }   mT  +    m'T \epsilon ( 1+ \epsilon )  -  \frac{1 }{2 } m'T \epsilon     \right].
\end{align}
On the other hand, the second term in \eqref{lemma2main}, $ \Delta(p_{0m'}, p_{mm'}, p_{m'm'})$, can be approximated as follows:
\begin{align}
 \Delta(p_{0m'}, p_{mm'}, p_{m'm'})=&\frac{MT}{4} \frac{ 2(p_{mm'}+p_{m'm'})^2 (1+p_{0m'}) +   p_{mm'}p_{m'm'}(1-p_{0m'})^2}{ (p_{mm'}+p_{m'm'})^2  (1-p_{0m'}) }
 \\\nonumber \approx &
 \frac{MT}{4} \frac{ 2\left(1-\frac{\epsilon}{P}\right)^2   +   \frac{1-\frac{\epsilon}{P}}{1+ \epsilon } \frac{\epsilon}{1+ \epsilon } \left(1-\frac{\epsilon}{P}\right) }{ (1-\frac{\epsilon}{P})^2  }
  \\\nonumber \approx &
 \frac{MT}{2}\left(1   +   \frac{\epsilon }{2(1+ \epsilon)^2 }   \right).
\end{align}

Therefore, ${\rm U}_{m'}$'s AoI can be approximated at high SNR as follows:
\begin{align}
\bar{\Delta}_{m'}^N =&\Delta_{m',0}  + \Delta(p_{0m'}, p_{mm'}, p_{m'm'})
\\\nonumber
\approx&   \frac{1}{(1+\epsilon)^2}  \left[   mT(1+\epsilon)   +  \frac{\epsilon }{2  }   mT  +    m'T \epsilon ( 1+ \epsilon )  -  \frac{1 }{2 } m'T \epsilon     \right] + \frac{1}{2}   MT  \left(1 +   \frac{\epsilon }{2(1+ \epsilon )^2}\right) \\\nonumber 
=&   \frac{1}{(1+\epsilon)}  \left[   mT      +    m'T \epsilon      \right] + \frac{1}{2}   MT  ,
\end{align}
where the last step follows by using the fact that $m'=m+\frac{M}{2}$.

Recall that ${\rm U}_{m'}$'s AoI for TDMA can be approximated at high SNR as follows: $m'T+\frac{MT}{2}$. Therefore,    the difference between ${\rm U}_{m'}$'s AoI for TDMA and NOMA, denoted by $D_{m'} $,  is given by
\begin{align}
D_{m'} \triangleq&   \frac{1}{(1+\epsilon)}  \left[   mT      +    m'T \epsilon      \right] + \frac{1}{2}   MT   - \left(m'T+\frac{MT}{2}\right)\\\nonumber
=& \frac{1}{(1+\epsilon)}  \left[   mT      +    m'T \epsilon      \right]    - m'T 
\\\nonumber
=& \frac{  mT      +    m'T \epsilon  -m'T -m'T \epsilon }{(1+\epsilon)}  =  -\frac{  MT   }{2(1+\epsilon)} <0,
\end{align}
which means that ${\rm U}_{m'}$ experiences less AoI for NOMA than for TDMA. 

In order to find a high-SNR approximation for $\bar{\Delta}_{m}^N$,     the following approximations can be obtained first: $\tau\approx  \frac{\epsilon }{P} $, $p_{mm} \approx 1-  \frac{\epsilon}{P}$,    $
p_{m'm} \approx      \frac{\epsilon}{P}      \frac{1-\frac{\epsilon}{P} }{1+ \epsilon } 
+ \frac{\epsilon }{P}  $, and $
p_{0m} 
\approx  \frac{\epsilon}{P}   \frac{\epsilon}{1+ \epsilon } 
   $.  By applying these   approximations and also using   steps  similar to those used for  approximating $\bar{\Delta}_{m'}^N$,   ${\rm U}_{m}$'s AoI can be approximated at high SNR as follows:
\begin{align}
\bar{\Delta}_{m}^N =&\Delta_{m,0}  + \Delta(p_{0m}, p_{mm}, p_{m'm}) \approx mT+ \frac{MT}{2} ,
\end{align}
which is identical to the case  of   TDMA. Therefore, the following corollary can be obtained.


\begin{corollary}\label{corollary2}
For the case with the GAR model, at high SNR, i.e.,  $P=P^S\rightarrow \infty$,   ${\rm U}_{m'}$'s   average AoI for CR-NOMA is strictly smaller than   that for TDMA. ${\rm U}_{m}$'s   average AoIs for TDMA and CR-NOMA are identical   at high SNR. 
\end{corollary} 

{\it Remark 3:} For ${\rm U}_{m'}$, the performance gain of NOMA over TDMA at high SNR is due to the fact that the use of NOMA reduces the access delay, which is particularly important in the GAR case. Recall that in the GAR case, a user's update is generated at the beginning of the TDMA frame, which means that a user which is scheduled later in the frame suffers from a larger AoI. The use of NOMA  always  ensures that ${\rm U}_{m'}$ does not have to wait until the $m'$-th time slot, but can transmit  earlier, i.e.,   in the $m$-th time slot, which leads to the AoI reduction   stated in Corollary \ref{corollary2}.  We   note that  the use of NOMA does not improve ${\rm U}_{m}$'s access delay, which is the reason why ${\rm U}_{m}$'s AoIs for TDMA and NOMA are identical  at high SNR. 

{\it Remark 4:}  We further note that the AoI reduction due to the use of CR-NOMA is   important to improve   user fairness. For example,   two user which are scheduled in the first and the last time slots of a frame experience significantly different AoI for TDMA, but the use of CR-NOMA can reduce the difference between the users' AoI experience.  

\section{Numerical Studies} \label{section 5}
In this section, the benefits of NOMA transmission regarding the AoI are studied   by using computer simulation results, where TDMA is used as a benchmarking scheme. The time frame structure shown in Fig.  \ref{fig1} is used, where the users' channel gains in different time slots are assumed to be i.i.d. complex Gaussian distributed with zero mean and unit variance. For   illustration,   define $R\triangleq\frac{N}{T} $ and assume that $P=P^S$, where  $P$  is termed the transmit SNR in the simulation section because the noise power is assumed to be normalized.  Because the AoI depends on the data generation models,   the benefits  of NOMA for     AoI reduction are studied in two different subsections in the following.

 \begin{figure}[t] \vspace{-0em}
\begin{center}
\subfigure[$ {R}=0.5$ bit/s/Hz]{\label{fig4a}\includegraphics[width=0.45\textwidth]{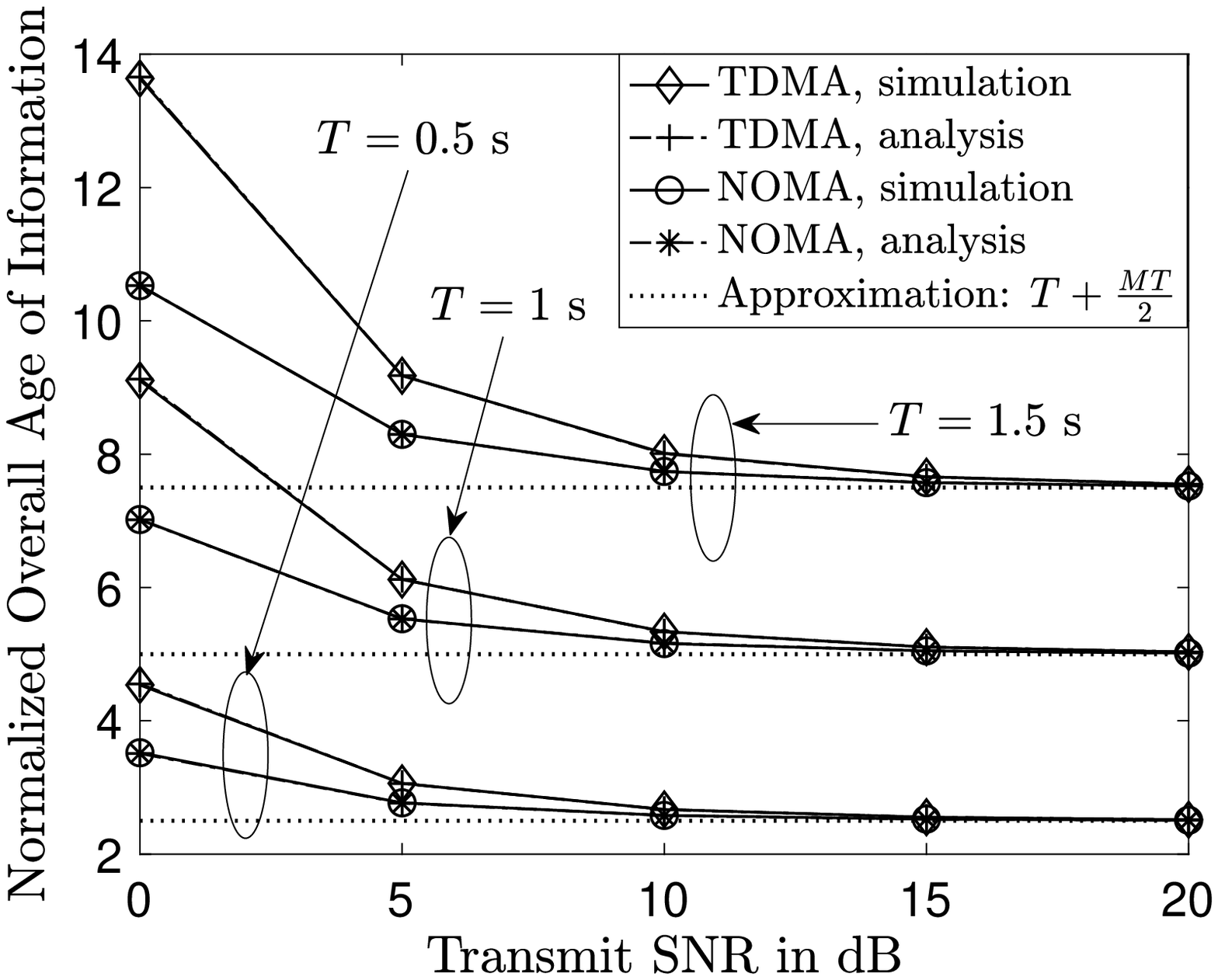}}\hspace{2em}
\subfigure[$ {R}=1$ bits/s/Hz]{\label{fig4b}\includegraphics[width=0.45\textwidth]{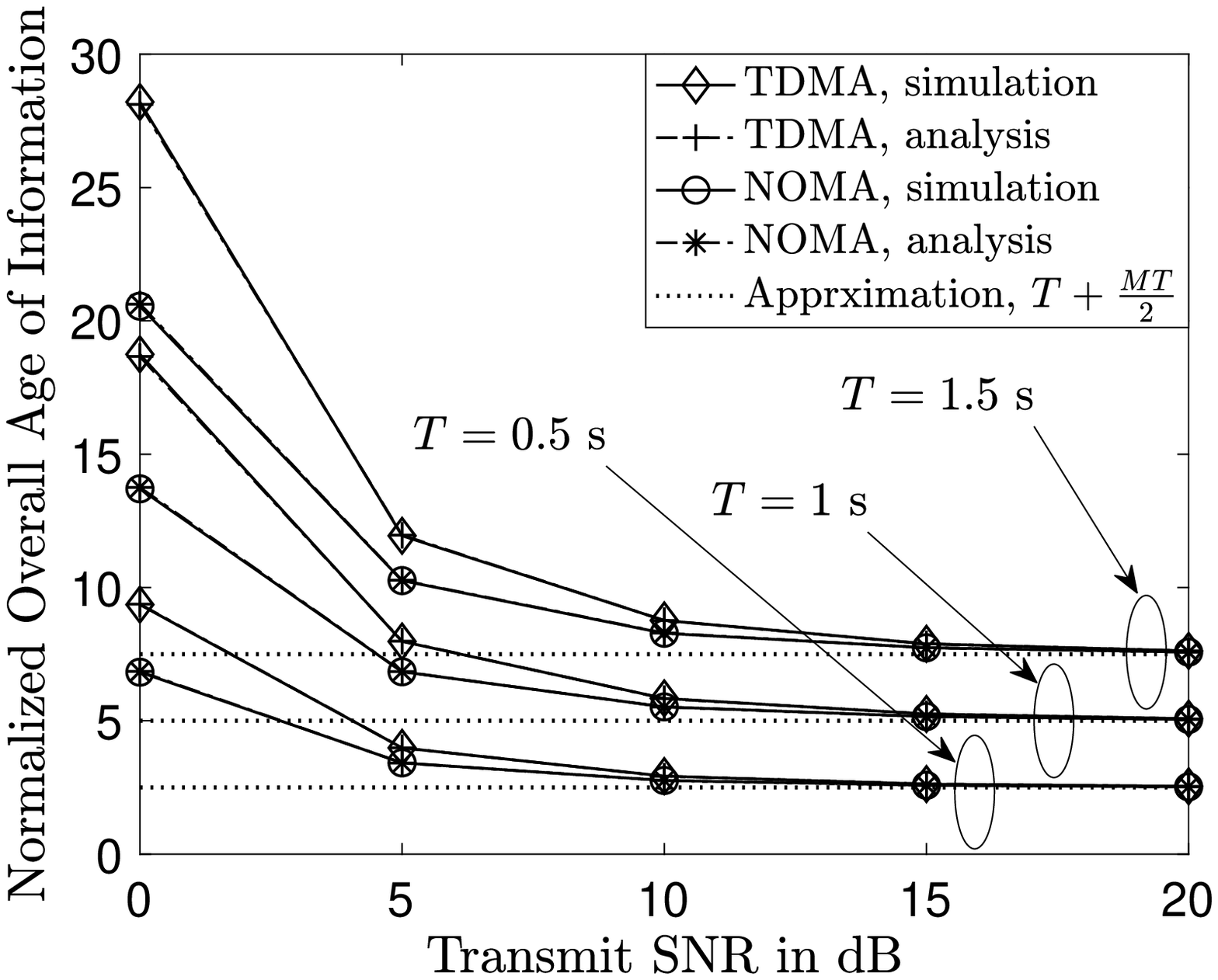}} \vspace{-1.5em}
\end{center}
\caption{ The impact of NOMA transmission on the AoI for  the generate-at-will model with  $M=8$.       \vspace{-1em} }\label{fig4}\vspace{-1em}
\end{figure}

\subsection{The Generate-at-Will Model} 
Recall that for the GAW model, each user generates a new update at the beginning of its transmit time slot. The impact of the different transmission protocols on the AoI is investigated in Fig. \ref{fig4} by assuming that there are $M=8$ users, i.e., there are $8$ time slots in each TDMA frame. As can be seen from the two subfigures in Fig. \ref{fig4},   the AoI achieved by the proposed NOMA transmission protocol can be significantly lower than that for   TDMA. For example, for     $R=1$ bits/s/Hz, $T=1.5$ s, and a transmit SNR of $0$ dB, the AoI realized with TDMA is around $28$, and the AoI achieved by CR-NOMA is just $20$, which means that the use of CR-NOMA yields more than a one-quarter   reduction   compared to TDMA. However, at high SNR, Fig. \ref{fig4} shows that TDMA and NOMA yield  the same   AoI, which confirms Corollary \ref{corollary1}.  We also note that both   subfigures in Fig. \ref{fig4}   verify the accuracy of the analytical results presented  in Lemma \ref{lemma1} as well as the approximation result  developed in \eqref{noma1}. 

Fig. \ref{fig4} also shows the impact of $R$ on the AoI performance of the considered transmission protocols. In particular, by comparing Fig. \ref{fig4a} and  Fig. \ref{fig4b}, one can observe that increasing $R$ increases the AoI for both   transmission protocols. Recall that increasing $R$ for given value   of $T$ means that there are more bits contained in each update, which makes transmission failures more likely and hence increases the   AoI.  We note that the performance gain of NOMA over TDMA becomes larger for larger  $R$. For example, for $T=1.5$ s and a transmit SNR of $0$ dB, the performance gain of NOMA over TDMA if $R=0.5$ bits/s/Hz is $3$, and this performance gain can be increased to almost $8$  if $R=1$ bits/s/Hz. The subfigures in Fig. \ref{fig4} also show that the AoI realized by the considered protocols     increases with  $T$, since increasing $T$ for given  $R$ means that there are more bits contained in each update.

    \begin{figure}[t]\centering \vspace{-0em}
    \epsfig{file=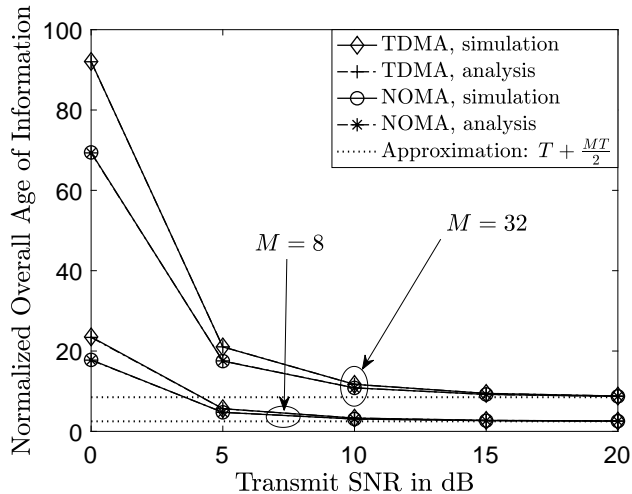, width=0.5\textwidth, clip=}\vspace{-0.5em}
\caption{The impact of the number of users, $M$, on the AoI, for  the generate-at-will model with  $R=1.5$ bits/s/Hz and  $T=0.5$~s.
  \vspace{-1em}    }\label{fig5}   \vspace{-0.1em} 
\end{figure}

In Fig. \ref{fig5}, the impact of the number of users, $M$, on the AoI achieved by the considered transmission protocols is studied. As can be seen from the figure, by increasing $M$, the AoI is increased for both   transmission protocols. This observation is expected since with more users in the network, each user has to wait for a longer period of time to be served. In addition, one can also observe that the performance gain of the proposed NOMA protocol over TDMA  increases as  the number of users, $M$, increases. For example, the performance gap between the two protocols is $5$ for   $M=8$, and   increases to $25$ for  $M=32$. This observation means that      the proposed NOMA protocol is particularly useful for reducing  the AoI of   networks with massive connectivity, which is a key use case of the 6G system. 

 \begin{figure}[t] \vspace{-0em}
\begin{center}
\subfigure[${\rm U}_m$'s AoI]{\label{fig6a}\includegraphics[width=0.45\textwidth]{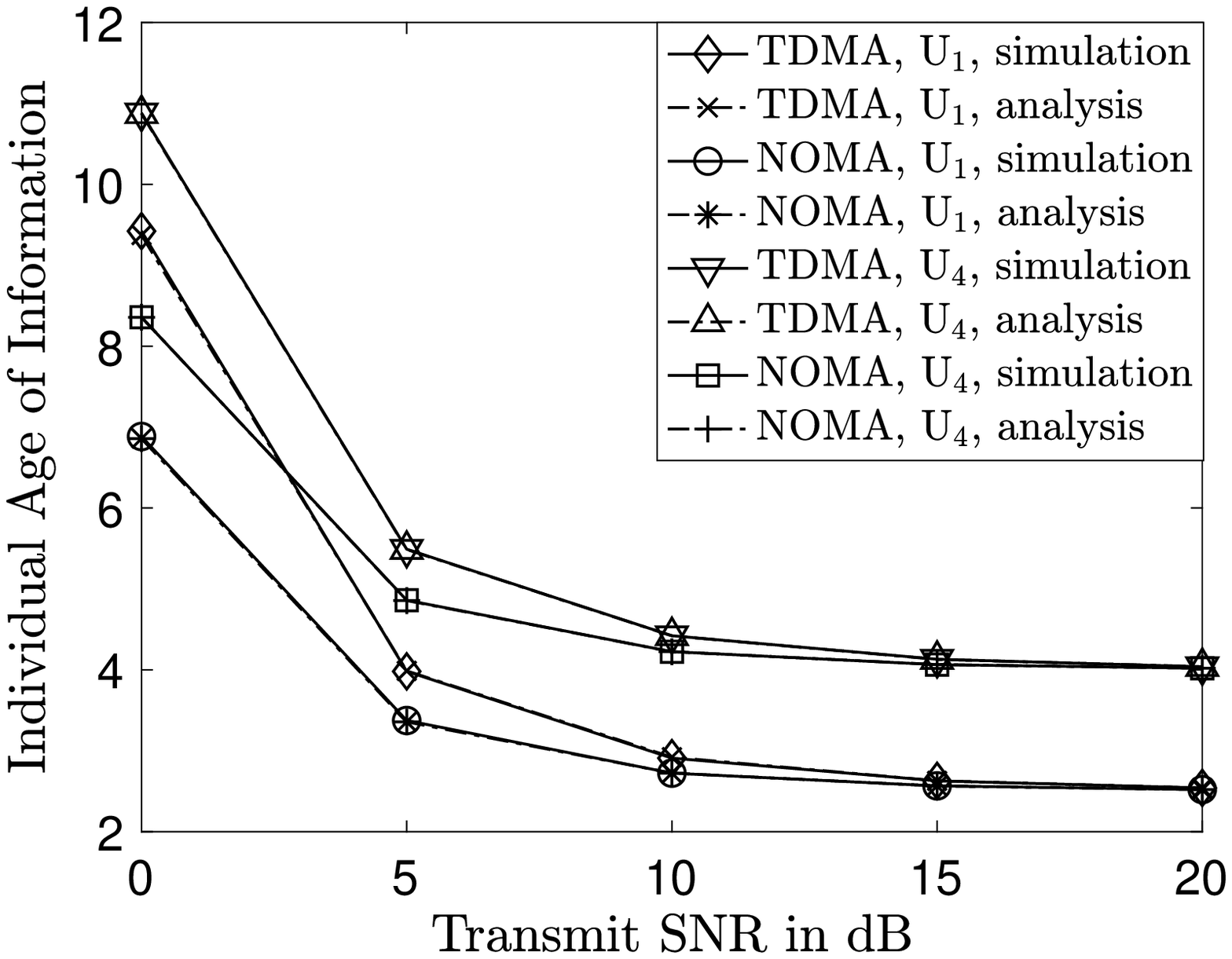}}\hspace{2em}
\subfigure[${\rm U}_{m'}$'s AoI]{\label{fig6b}\includegraphics[width=0.45\textwidth]{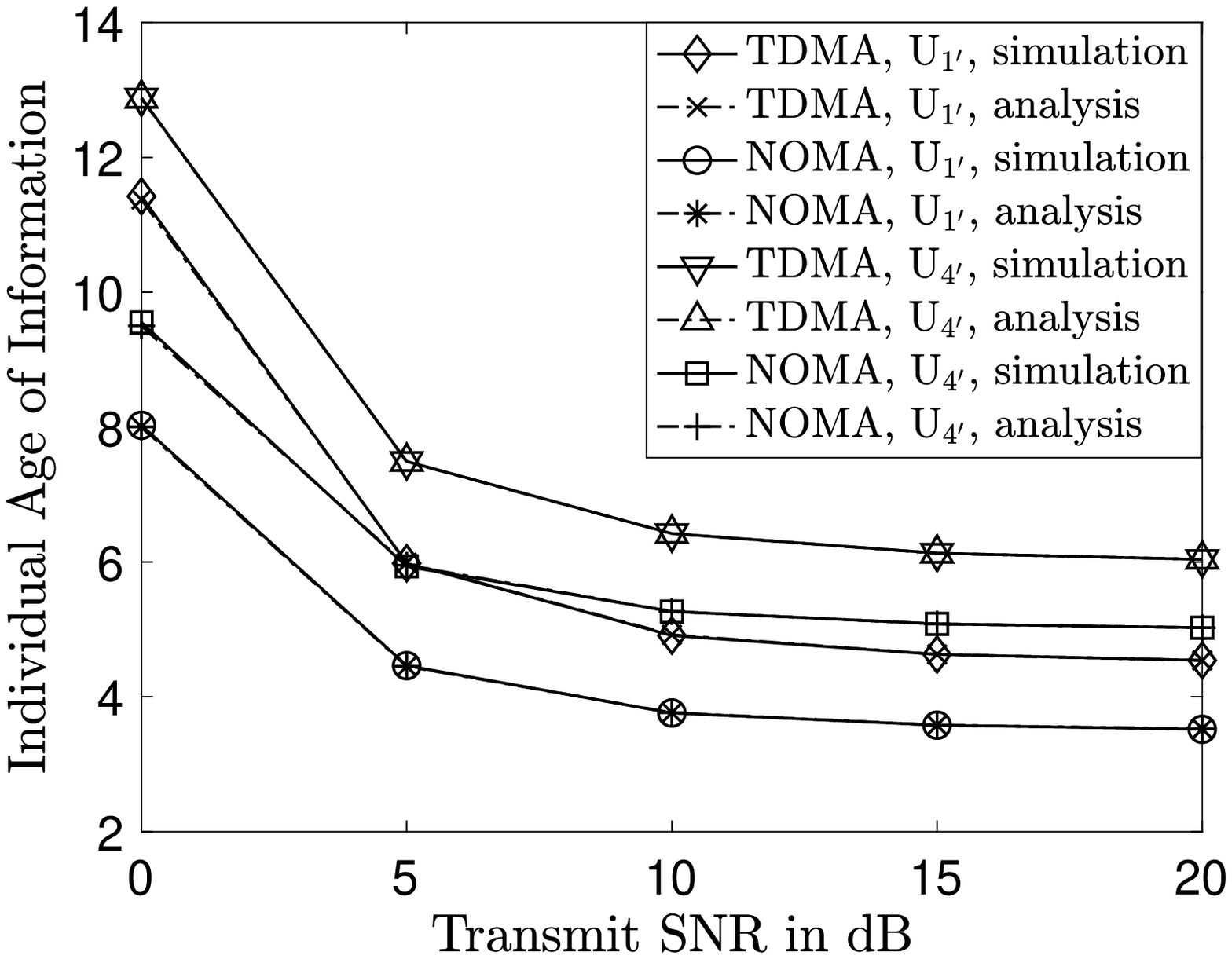}} \vspace{-1.5em}
\end{center}
\caption{ The impact of NOMA transmission on the individual AoI for the generate-at-request model with    $M=8$, $R=1$ bit/s/Hz and $T=0.5$ s.     \vspace{-1em} }\label{fig6}\vspace{-1em}
\end{figure}

\subsection{The Generate-at-Request Model}
Recall that for the GAR model, each user generates its update at the beginning of each TDMA frame, instead of each time slot as in the GAW case. 
 In Fig. \ref{fig6}, the impact of the NOMA transmission protocol on the users' individual AoI is studied for the GAR model. In particular, Fig. \ref{fig6a} focuses on  ${\rm U}_m$'s  individual   AoI achieved by the two considered  transmission protocols,   $1\leq m \leq \frac{M}{2}$. Unlike for the GAW model, different users experience different AoIs for the GAR model, i.e., ${\rm U}_m$'s AoI is larger than that of ${\rm U}_i$'s, $m>i$.  This observation is expected since  ${\rm U}_m$'s instantaneous AoI can   drop  to $mT$ at most, whereas ${\rm U}_i$'s instantaneous AoI can   drop  to $iT$. Fig. \ref{fig6a}  also shows that for ${\rm U}_m$, the performance gain of NOMA over TDMA is similar to that for the GRW case, e.g., the use of NOMA yields a significant performance gain at low SNR but achieves  the same AoI as TDMA at high SNR. This performance gain at low SNR is due to the fact that ${\rm U}_m$ has a second chance for transmission in each frame, whereas for TDMA, ${\rm U}_m$ has to rely on a single time slot for its updates. 

Fig. \ref{fig6b} focuses on  ${\rm U}_{m'}$'s individual   AoI achieved for the two transmission protocols.  As can be seen from the figure, NOMA   outperforms TDMA in all SNR regimes. The reason for this superior performance   can be explained as follows.  Recall that for TDMA, ${\rm U}_{m'}$ has to rely on the $m'$-th time slot only for sending its update to the base station. The use of the proposed NOMA protocol has  two advantages for reducing the AoI. One is that  the use of NOMA offers the user two chances to transmit  in each TDMA time frame. The other is that the proposed NOMA protocol can schedule ${\rm U}_{m'}$ to transmit earlier, i.e.,  completing its update in the $m$-th time slot, instead of waiting for the $m'$-th time slot as for TDMA. The latter is crucial    for NOMA to outperform   TDMA in the high SNR regime, as indicated by  Corollary \ref{corollary2}. Furthermore, we note that the subfigures of Fig. \ref{fig6} demonstrate the accuracy of the developed analytical results shown in Lemma \ref{lemma2}.

    \begin{figure}[t]\centering \vspace{-0em}
    \epsfig{file=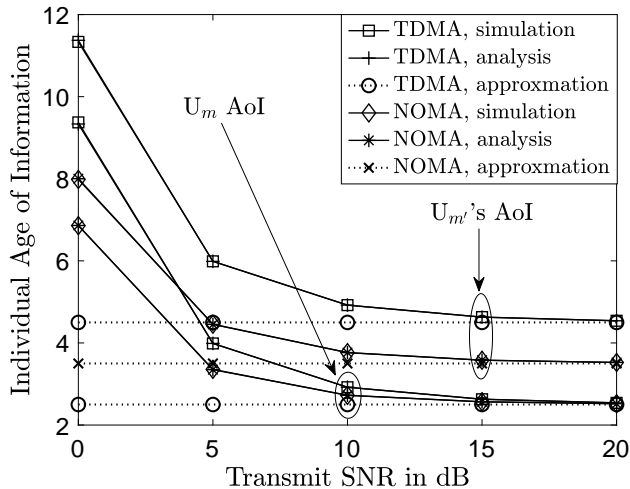, width=0.5\textwidth, clip=}\vspace{-0.5em}
\caption{Comparison of the AoI experienced by  ${\rm U}_m$  and  ${\rm U}_{m'}$,  where the generate-at-request model is used, $m=1$,   $M=8$, $R=1$ bit/s/Hz and $T=0.5$ s.  
  \vspace{-1em}    }\label{fig7x}   \vspace{-0.1em} 
\end{figure}

In Fig. \ref{fig7x}, the individual  AoI experienced by  ${\rm U}_m$  and  ${\rm U}_{m'}$ is compared.  As can be seen from the figure, for   the TDMA     transmission protocol, ${\rm U}_{m'}$'s AoI is much larger than that of ${\rm U}_{m}$, which is due to the fact that ${\rm U}_{m'}$ has to wait for the $m'$-th time slot to deliver its update and hence experiences large access delays. By using the proposed NOMA protocol, the difference between the two users' AoI can be reduced significantly. For example, at high SNR, the difference between the two users' AoI is $2$ with TDMA, and can be halved by applying NOMA. Therefore, the use of NOMA can effectively reduce the difference between the users' AoI  and hence  improve  user fairness, as discussed in Remark 4. We also  note that Fig. \ref{fig7x}   demonstrates the accuracy of the high SNR approximations   developed in Section  \ref{subsection IV.B}. 

 \begin{figure}[t] \vspace{-0em}
\begin{center}
\subfigure[$R=0.5$ bits/s/Hz]{\label{fig7a}\includegraphics[width=0.45\textwidth]{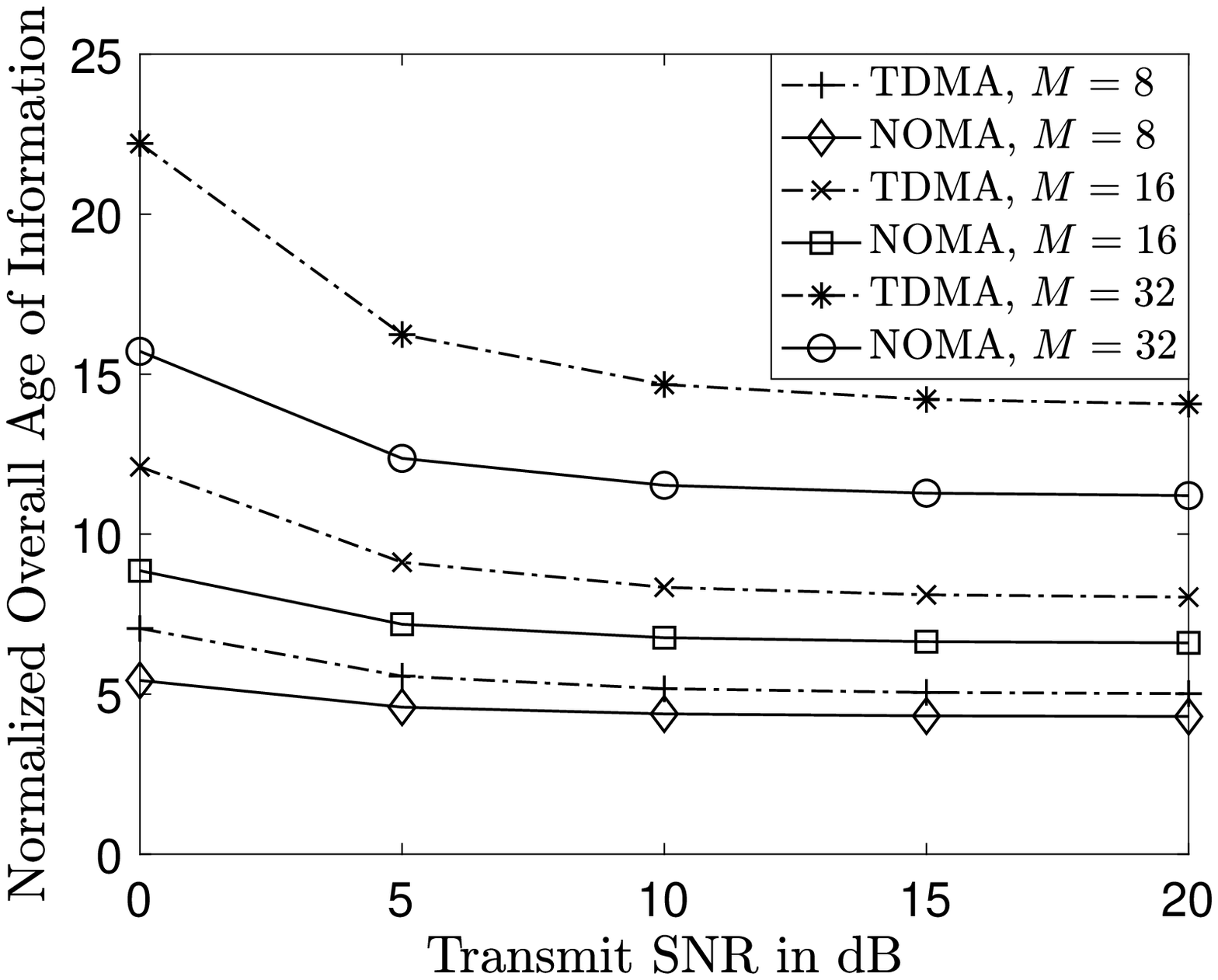}}\hspace{2em}
\subfigure[$ {R}=1.5$ bits/s/Hz]{\label{fig7b}\includegraphics[width=0.45\textwidth]{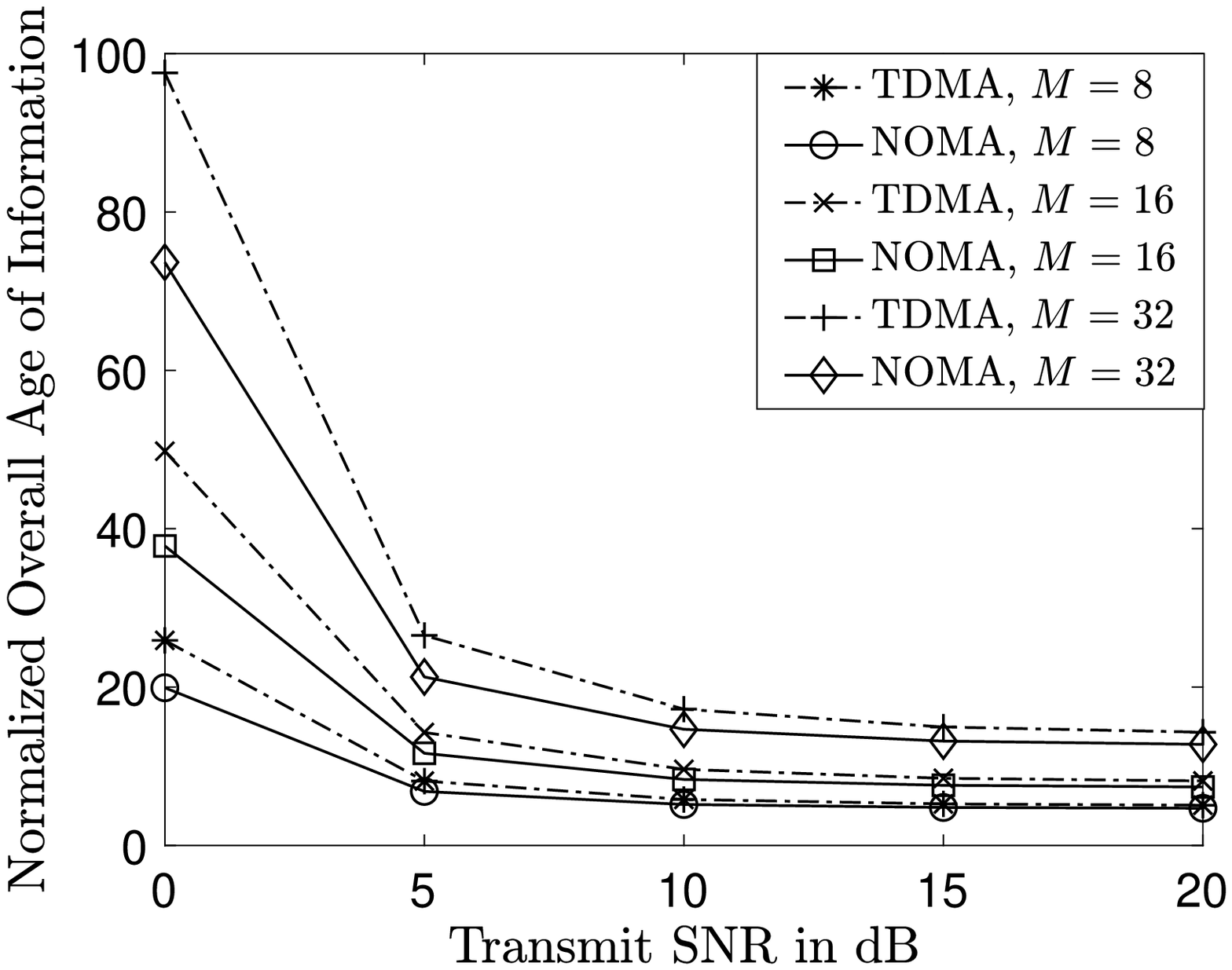}} \vspace{-1.5em}
\end{center}
\caption{The impact of NOMA transmission on the overall AoI, where the generate-at-request model is used and     $T=0.5$ s.    \vspace{-1em} }\label{fig7}\vspace{-1em}
\end{figure}

In Fig. \ref{fig7}, the normalized overall AoI is used as the metric to study the performance  of the proposed NOMA transmission protocol with the GAR model. 
We note that    for   the  GAW model,  the proposed   NOMA protocol has the limitation that  it can outperform TDMA in the low SNR regime only, as shown in Figs. \ref{fig4} and \ref{fig5}.  Compared to  Figs. \ref{fig4} and \ref{fig5},  Fig. \ref{fig7} shows that the proposed NOMA protocol can always outperform TDMA and realize  a smaller  overall AoI in all   SNR regimes. The performance gain of NOMA over TDMA is significant for  small $R$, and is reduced by increasing $R$, as shown in the two subfigures of Fig. \ref{fig7}. Furthermore, Fig.~\ref{fig7} also demonstrates that the performance gain of NOMA over TDMA increases as the number of the users, $M$, increases, which is consistent with the observations related to Fig. \ref{fig5}.  

\section{Conclusions} \label{section 6}
In this paper, NOMA has been used as an add-on to  reduce the AoI of a legacy TDMA network. By using the key features of the two considered    data generation models, namely GAW and GAR, two CR-NOMA transmission protocols have been developed to reduce the AoI of the network.   Closed-form expressions of  the AoI achieved by the proposed NOMA protocols have been derived, and an asymptotic analysis  has been carried out to show that the use of NOMA can reduce the AoI due to the following two reasons. First,  the use of NOMA provides users more chances to transmit, which ensures that the users can  update their base station more frequently.  Second,   the use of NOMA allows  the users to transmit earlier than in the TDMA case, and hence, improve the freshness of the data available at the base station.
In this paper, it was assumed that the signals of the secondary users   are decoded first at the base station before decoding the primary users' signals. The use of more dynamic SIC decoding orders can potentially lead to a larger  AoI reduction, which is  an important direction for future research. In addition, in this paper, the users' CSI is assumed to be perfectly known  for the implementation of NOMA. An important future research direction is to study the impact of imperfect CSI  on the AoI reduction.

\appendices
\section{Proof for Lemma \ref{lemma1}}\label{proof1}
It is straightforward to verify that all the users experience the same AoI with CR-NOMA for the GAW case, and therefore,   ${\rm U}_m$'s AoI  is considered  in the remainder of this proof.  
Unlike the case of TDMA, the duration between the beginning and the end of the $j$-th successful update is not always a multiple   of $MT$, since each user has two chances to transmit  in each frame.  
For illustration, assume that ${\rm U}_m$'s  $(j-1)$-th successful  update finishes  in the $i$-th frame. Because ${\rm U}_m$ has two chances to transmit in each frame, the following two events are defined based on which of the two time slots  is used:
\begin{align}\nonumber
E_{jm} =&  \left\{\text{the $(j-1)$-th successful  update finishes at the end  } \right.\\\nonumber & \text{ of the $m$-th time slot of a frame, e.g., at $t_i^{m+1}$}  \},\\\nonumber
E_{jm'} =&   \left\{\text{the $(j-1)$-th successful  update finishes at the end } \right.\\ \label{Ejm} &  \text{ of the $m'$-th time slot of a frame, e.g., at $t_i^{m'+1}$}  \}.
\end{align}

 Denote the time interval  between   the  $(j-1)$-th and the $j$-th successful updates by $y_j$, $j\geq 1$, whose value can be obtained considering  by the following four cases:
\begin{align} \label{atwill yj}
y_j = \left\{\begin{array}{ll}\hspace{-0.5em}
x_jMT,&\hspace{-0.5em}\text{from $t_i^{m+1}$    to $x_jMT +t_i^{m+1} $ } 
\\\hspace{-0.5em}
x_jMT,&\hspace{-0.5em}\text{from  $t_i^{m'+1}$ to $x_jMT + t_i^{m'+1}$ } 
\\
\hspace{-0.5em}x_jMT+\frac{M}{2}T,&\hspace{-0.5em}\text{from $t_i^{m+1}$ to  $x_jMT +\frac{M}{2}T+t_i^{m+1}$ } 
\\
\hspace{-0.5em}x_jMT-\frac{M}{2}T,&\hspace{-0.5em}\text{from  $t_i^{m'+1}$ to  $x_jMT -\frac{M}{2}T+t_i^{m'+1}$ } 
\end{array}\right.
\hspace{-0.5em},
\end{align}
where   $x_j$ is defined as the number of frames between the $(j-1)$-th and the $j$-th successful updates, $x_j\in\mathbb{Z}$, and $\mathbb{Z}$ denotes the integer set. 

Eq. \eqref{atwill yj} shows  that   $y_j=x_jMT$ is caused by  two different events. One is that, conditioned on $E_{jm}$,  the user   fails to update the base station during the first $(x_j-1)$ frames, but successfully sends an update in the $m$-th time slot of the $(x_j+i)$-th frame, where the user's $(j-1)$-th successful update is assumed to occur   in the $i$-th frame without loss of generality. The other is that, conditioned on $E_{jm'}$, the user   fails to update the base station during the first $(x_j-1)$ frames, but successfully sends an update in the $m'$-th time slot of the $(x_j+i)$-th frame. 
  $y_j=x_jMT+\frac{M}{2}T$ corresponds   to the event that, conditioned on $E_{jm}$, the user fails to update the base station until  the $m'$-th time slot of the $(x_j+i)$-th frame. $y_j=x_jMT-\frac{M}{2}T$ corresponds to   the event that, conditioned on $E_{jm'}$,  the user fails to update the base station until  the $m$-th time slot of the $(x_j+i)$-th frame. 
Following the definition of $y_j$ in \eqref{atwill yj},    the  following four conditional probabilities can be defined: $p_{j1}=\mathbb{P}(y_j=x_jMT|E_{jm})$,  $p_{j2}=\mathbb{P}(y_j=x_jMT+\frac{M}{2}T|E_{jm})$,  $p_{j3}=\mathbb{P}(y_j=x_jMT|E_{jm'})$, and    $p_{j4}=\mathbb{P}(y_j=x_jMT-\frac{M}{2}T|E_{jm'})$, which will be evaluated later.

 ${\rm U}_m$'s average AoI achieved by NOMA can be expressed as follows: 
\begin{align}
\bar{\Delta}_m^N =& \underset{J\rightarrow \infty}{\lim} \frac{\sum^{J}_{j=1} Q_j }{\sum^{J}_{j=1}y_j}  
\\\nonumber =& \underset{J\rightarrow \infty}{\lim} \frac{\sum^{J}_{j=1} Ty_j+\frac{1}{2}y_j^2 }{\sum^{J}_{j=1}y_j} =T+ \frac{1}{2} \frac{  \mathcal{E}\{Y^2\}}{ \mathcal{E}\{Y\}},
\end{align}
 where $\mathcal{E}\{Y\}=\underset{J\rightarrow \infty}{\lim}\frac{1}{J}\sum^{J}_{j=1}y_j$ and $\mathcal{E}\{Y^2\}=\underset{J\rightarrow \infty}{\lim}\frac{1}{J}\sum^{J}_{j=1}y_j^2$. The remainder of the proof is to evaluate $\mathcal{E}\{Y\}$ and $\mathcal{E}\{Y^2\}$. 

Define $p_0$ as the probability of the event that ${\rm U}_m$ fails to deliver an update in both of the two   time slots in  one frame.  Note that, in  the two times slots, ${\rm U}_m$ assumes different roles for transmission, i.e., ${\rm U}_m$ is the primary user in the $m$-th time slot and the secondary user in the $m'$-th time slot. By using the date rate constraint in \eqref{crnoma} and  the assumption that    the users' channel gains are i.i.d. complex Gaussian distributed, the probability, $p_0$, can  be obtained as follows: 
\begin{align}\nonumber
p_0 =  & \mathbb{P}\left( \log\left( 1+ P|h_m^{i,m}|^2 \right)\leq \frac{N}{T},   \log\left( 1+\frac{P^S|h_m^{i,m'}|^2}{P|h_{m'}^{i,m'}|^2+1}\right)\leq \frac{N}{T} \right) = \left(1-e^{- \frac{\epsilon}{P^S}}\right)  \left(1- \frac{e^{-\frac{\epsilon}{P^S} }}{1+\frac{P\epsilon}{P^S}}\right) ,
\end{align}
where  the channel gains in the $i$-th frame are used for illustrative purposes.
 
Denote the probability for the user to successfully deliver its update in the $m$-th time slot by $p_m $, which can be expressed as follows:
\begin{align}
p_m =  & \mathbb{P}\left( \log\left( 1+ P^S|h_m^{i,m}|^2 \right)\geq \frac{N}{T}  \right)  = e^{-\frac{\epsilon}{P^S}}.
\end{align}

 Denote the probability for the event that the user fails to    deliver its update in the $m$-th time slot   but successfully delivers a new update in the $m'$-th time slot by $p_{m'} $. This probability can be expressed as follows:
 \begin{align}\nonumber
p_{m'} =  &  \mathbb{P}\left( \log\left( 1+ P|h_m^{i,m}|^2 \right)\leq \frac{N}{T},   \log\left( 1+\frac{P^S|h_m^{i,m'}|^2}{P|h_{m'}^{i,m'}|^2+1}\right)\geq \frac{N}{T} \right) \\\nonumber =&\left(1-e^{- \frac{\epsilon}{P^S}}\right)   e^{-\frac{\epsilon}{P^S} }\frac{1}{1+\frac{P\epsilon}{P^S}}   .
\end{align}

Therefore, $\mathcal{E}\{Y\} $ can be obtained as follows:
\begin{align}
\mathcal{E}\{Y\} =& \underset{x_j\in\mathbb{Z}}{\sum}x_jMT  \left(p_{j1}\mathbb{P}(E_{jm}) +p_{j3}\mathbb{P}(E_{jm'})\right) +\underset{x_j\in\mathbb{Z}}{\sum} \left(x_jMT+\frac{M}{2}T\right)  p_{j2}\mathbb{P}(E_{jm})\\\nonumber &
+ \underset{x_j\in\mathbb{Z}}{\sum} \left(x_jMT-\frac{M}{2}T\right)  p_{j4}\mathbb{P}(E_{jm'})\\\nonumber
 =& \sum^{\infty}_{j=1} jMT\left( p_0^{j-1}p_m^2+p_0^{j-1}p_{m'}^2\right) + \sum^{\infty}_{j=1} \left(jMT+\frac{M}{2}T\right)p_0^{j-1}p_{m'}p_m\\\nonumber &+ \sum^{\infty}_{j=1} \left(jMT-\frac{M}{2}T\right)p_0^{j-1}p_{m}p_{m'} ,
\end{align}
where the last step is obtained by using the following facts: $\mathbb{P}(E_{jm})=p_{m}$, $\mathbb{P}(E_{jm'})=p_{m'}$, $p_{j1}=p_0^{j-1}p_m$, $p_{j2}=p_0^{j-1}p_{m'}$,  $p_{j3}=p_0^{j-1}p_{m'}$, and $p_{j4}=p_0^{j-1}p_{m}$.
With some straightforward  algebraic manipulations, the expression of $\mathcal{E}\{Y\}  $ can be simplified as follows:
\begin{align}
\mathcal{E}\{Y\}  
 =& MT(p_m^2+p_{m'}^2+2p_{m}p_{m'})p_0^{-1}\sum^{\infty}_{j=1} jp_0^{j} 
  \\\nonumber 
 =& MT(p_m+p_{m'})^2  \frac{1}{(1-p_0)^2}  ,
\end{align}
where the last step follows from the following infinite sum of series:  
\begin{align}\label{sum1}
\sum^{\infty}_{j=1} jx^{j}&=x\sum^{\infty}_{j=1} \frac{d}{dx}x^{j}=x \frac{d}{dx}\sum^{\infty}_{j=1}x^{j}= \frac{x}{(1-x)^2}.
\end{align}
for $0<x<1$. 

On the other hand,   $\mathcal{E}\{Y^2\} $ can be obtained as follows:
\begin{align}
\mathcal{E}\{Y^2\} =& \sum^{\infty}_{j=1} j^2M^2T^2 \left(p_{j1}\mathbb{P}(E_{jm}) +p_{j3}\mathbb{P}(E_{jm'})\right)  + \sum^{\infty}_{j=1} \left(jMT+\frac{M}{2}T\right) ^2 p_{j2}\mathbb{P}(E_{jm}) 
\\\nonumber &+ \sum^{\infty}_{j=1} \left(jMT-\frac{M}{2}T\right) ^2 p_{j4}\mathbb{P}(E_{jm'}) 
\\\nonumber
=& \sum^{\infty}_{j=1} j^2M^2T^2\left( p_0^{j-1}p_m^2+p_0^{j-1}p_{m'}^2\right) + \sum^{\infty}_{j=1} \left(jMT+\frac{M}{2}T\right) ^2 p_0^{j-1}p_{m'}p_m
\\\nonumber &+ \sum^{\infty}_{j=1} \left(jMT-\frac{M}{2}T\right) ^2 p_0^{j-1}p_{m}p_{m'} .
\end{align}
With some straightforward  algebraic manipulations, the expression of $\mathcal{E}\{Y^2\}  $ can be simplified as follows:
\begin{align}
\mathcal{E}\{Y^2\} =&M^2T^2(p_m+p_{m'})^2p_0^{-1} \sum^{\infty}_{j=1} j^2p_0^{j} +  \frac{M^2}{2}T^2p_mp_{m'}p_0^{-1}\sum^{\infty}_{j=1} p_0^{j} 
\\\nonumber 
=&M^2T^2(p_m+p_{m'})^2 \frac{(1+p_0)}{(1-p_0)^3} +  \frac{M^2}{2}T^2\frac{p_mp_{m'}}{1-p_0} ,
\end{align}
where the last step follows from the following infinite sum of series:  
\begin{align}\label{sum2}
\sum^{\infty}_{j=1} j^2x^{j} =& x\sum^{\infty}_{j=1} j  \frac{d}{dx}x^{j} = x \frac{d}{dx}\sum^{\infty}_{j=1} j x^{j}   =\frac{x(1+x)}{(1-x)^3}.
\end{align}

Therefore, ${\rm U}_m$'s average AoI can be calculated as follows:
\begin{align}
\bar{\Delta}_m^N  =&T+ \frac{1}{2} \frac{  \mathcal{E}\{Y^2\}}{ \mathcal{E}\{Y\}} \\\nonumber 
=&T+ \frac{1}{2} \frac{ M^2T^2(p_m+p_{m'})^2 \frac{(1+p_0)}{(1-p_0)^3}+  \frac{M^2}{2}T^2p_mp_{m'}\frac{1}{1-p_0} }{ MT(p_m+p_{m'})^2  \frac{1}{(1-p_0)^2} }.
\end{align}
By using the fact that all the users experiences the same AoI and with some straightforward algebraic manipulations, the lemma is proved.

 \section{Proof for Lemma \ref{lemma2}}\label{proof3}
For the GAR model, different users experience different AoI.s In the proof, the common steps for the analysis of the users' AoIs are provided first and then the specific results for the  users'  individual AoIs are presented.  

\subsection{Generic Expression for the AoI, $\bar{\Delta}_k^N $, $k\in\{m,m'\}$}
To facilitate the analysis of the AoI,  the events in \eqref{Ejm} are first modified as follows: 
\begin{align}\nonumber
E_{jm}^k =&  \left\{\text{${\rm U}_k$'s   $(j-1)$-th successful  update finishes at the} \right.\\\nonumber & \text{end of the $m$-th time slot of a frame, e.g., at $t_i^{m+1}$}  \},\\\nonumber
E_{jm'}^k =&   \left\{\text{${\rm U}_k$'s   $(j-1)$-th successful  update finishes at the} \right.\\ \label{Ejmx} &  \text{end of the $m'$-th time slot of a frame, e.g., at $t_i^{m'+1}$}  \}.
\end{align}
where $k\in\{m,m'\}$.

Denote by $y_{jk}$   the time interval between ${\rm U}_k$'s $(j-1)$-th and   $j$-th successful updates, $k\in\{m,m'\}$. Depending on which of the two events, $E_{jm}^k$ and $E_{jm'}^k$, happens, the value of $y_{jk}$ will be different. Furthermore, the height of the rectangle in the shaded region shown in Fig. \ref{fig3}   also depends on the two events, $E_{jm}^k$ and $E_{jm'}^k$.     For example,   ${\rm U}_m$'s instantaneous AoI is reset to $mT$ if the user's $(j-1)$-th successful update finishes in the $m$-th time slot of the last frame, i.e.,   $E^m_{jm}$ occures. If $E^m_{jm'}$ occures, i.e., the user's $(j-1)$-th successful update finishes in the $m'$-th time slot of the last frame,  ${\rm U}_m$'s  AoI is reset to $m'T$.  ${\rm U}_{m'}$'s instantaneous AoI is changed similar to that of ${\rm U}_m$'s  AoI.  Therefore, the average AoI achieved by CR-NOMA can be expressed as follows:   
\begin{align}
\bar{\Delta}_k^N =& \underset{J\rightarrow \infty}{\lim} \frac{\sum^{J}_{j=1} Q_j^k }{\sum^{J}_{j=1}y_{jk}}  \\\nonumber
= &\underset{J\rightarrow \infty}{\lim} \frac{\sum^{J}_{j=1} \left(\mathbf{1}_{E^k_{jm}}mT + \mathbf{1}_{E^k_{jm'}}m'T \right)y_{jk}+\frac{1}{2}y_{jk}^2 }{\sum^{J}_{j=1}y_{jk}} 
,
\end{align}
where $k\in\{m, m'\}$,   $Q_j^K$ denotes the area of the shaded shape   shown in Fig. \ref{fig3}, and $ \mathbf{1}_{E}$ is an indicator  function, i.e.,  $\mathbf{1}_{E}=1$ if   event $E$ happens, otherwise $\mathbf{1}_{E}=0$. By using   steps similar to those in the proof of Lemma \ref{lemma1},  $\bar{\Delta}_k^N $ can be expressed as follows:
\begin{align}\label{delta request}
\bar{\Delta}_k^N  =&\Delta_{k,0}+ \frac{1}{2} \frac{  \mathcal{E}\{Y_k^2\}}{ \mathcal{E}\{Y_k\}},
\end{align}
where  $\mathcal{E}\{Y_k\}=\underset{J\rightarrow \infty}{\lim}\frac{1}{J}\sum^{J}_{j=1}y_{jk}$,  $\mathcal{E}\{Y_k^2\}=\underset{J\rightarrow \infty}{\lim}\frac{1}{J}\sum^{J}_{j=1}y_{jk}^2$, and $\Delta_{k,0}$ is defined   as follows: 
\begin{align}
\Delta_{k,0}\triangleq &\underset{J\rightarrow \infty}{\lim} \frac{\sum^{J}_{j=1} \left(\mathbf{1}_{E^k_{jm}}mT + \mathbf{1}_{E^k_{jm'}}m'T \right)y_{jk}  }{\sum^{J}_{j=1}y_{jk}} .
\end{align}
  To facilitate the performance analysis, denote by $ {p}_{0k}$     the probability of the event that ${\rm U}_k$ fails to deliver an update in both of the two time slots of one frame,   by $p_{mk}$      the probability of the event that ${\rm U}_k$ successfully delivers its update in the $m$-th time slot of a frame, and  by  $p_{m'k}$     the probability of the event that ${\rm U}_k$ fails  in the $m$-th time slot but successfully delivers its update in the $m'$-th time slot of the same frame, $k\in\{m,m'\}$. 
  
  By using the same steps in Appendix  \ref{proof1}, it is straightforward to show that  the second term in \eqref{delta request} is simply $ \Delta(p_{0k}, p_{mk}, p_{m'k})$. Therefore, in the remainder of the proof, we  focus on the evaluation of   the first term in \eqref{delta request},  $\Delta_{k,0}  $, as well as the probabilities, $p_{0k}$, $p_{mk}$ and $p_{m'k}$, as shown in the following sections.

\subsection{Evaluation of  $\Delta_{k,0}  $, $k\in\{m,m'\}$}

By using the expectation   $\mathcal{E}\{Y_k\}$,  $\Delta_{k,0}$ in \eqref{delta request} can be expressed as follows:
\begin{align}
\Delta_{k,0}=&\underset{J\rightarrow \infty}{\lim}  \frac{\frac{\sum^{J}_{j=1}\mathbf{1}_{E^k_{jm}}mTy_{jk} }{J}  + \frac{ \sum^{J}_{j=1}\mathbf{1}_{E^k_{jm'}}m'T y_{jk}}{J}  }{\frac{\sum^{J}_{j=1}y_{jk}}{J}} \\\nonumber 
=&\frac{mT}{{\mathcal{E}\{Y_k\}} }
\underset{J\rightarrow \infty}{\lim}  \frac{\sum^{J}_{j=1}\mathbf{1}_{E^k_{jm}}}{J}\frac{\sum^{J}_{j=1}\mathbf{1}_{E^k_{jm}}y_{jk}}{\sum^{J}_{j=1}\mathbf{1}_{E^k_{jm}}}  +\frac{m'T}{{\mathcal{E}\{Y_k\}} }
\underset{J\rightarrow \infty}{\lim}  \frac{\sum^{J}_{j=1}\mathbf{1}_{E^k_{jm'}}}{J} \frac{ \sum^{J}_{j=1}\mathbf{1}_{E^k_{jm'}} y_{jk}}{\sum^{J}_{j=1}\mathbf{1}_{E^k_{jm'}}}  .
\end{align}

Define the following two conditional expectations: $\mathcal{E}\{Y_k|E^k_{jm}\}=\underset{J\rightarrow \infty}{\lim} \frac{\sum^{J}_{j=1}\mathbf{1}_{E^k_{jm}}mTy_{jk}}{\sum^{J}_{j=1}\mathbf{1}_{E^k_{jm}}} $  and $\mathcal{E}\{Y_k|E^k_{jm'}\}=\underset{J\rightarrow \infty}{\lim} \frac{ \sum^{J}_{j=1}\mathbf{1}_{E^k_{jm'}}m'T y_{jk}}{\sum^{J}_{j=1}\mathbf{1}_{E^k_{jm'}}}$, which can be used to simplify the expression for $\Delta_{k,0}  $ as follows:
\begin{align}
\Delta_{k,0}  
=&
 \frac{mT\mathbb{P}(E^k_{jm})\mathcal{E}\{Y_k|E^k_{jm}\}   + m'T\mathbb{P}(E^k_{jm'})\mathcal{E}\{Y_k|E^k_{jm'}\}  }{\mathcal{E}\{Y_k\}} .
\end{align}

For illustrative purposes, assume that  ${\rm U}_k$'s $(j-1)$-th successful update happens in the $i$-th frame.   Therefore,  $ E^k_{jm}$  means that ${\rm U}_k$'s $(j-1)$-th successful update happens in the $m$-th time slot of the $i$-th frame, and hence  $y_{jk}$ can have the following two forms: 
\begin{align} \label{atwill yj2}
y_{jk} = \left\{\begin{array}{ll} 
x_{jk}MT,& \text{from  $t_i^{m+1}$ to $x_{jk}MT +t_i^{m+1} $ } 
\\ 
x_{jk}MT+\frac{M}{2}T,& \text{from  $t_i^{m+1}$ to $x_{jk}MT +\frac{M}{2}T+t_i^{m+1}$ } 
\end{array}\right.\hspace{-1em},
\end{align}
where $x_{jk}$ denotes the number of frames between ${\rm U}_k$'s $(j-1)$-th and $j$-th successful updates. 
 Therefore, 
$\mathcal{E}\{Y_k|E^k_{jm}\} $ can be obtained as follows:
\begin{align}
\mathcal{E}\{Y_k|E^k_{jm}\} =& \underset{x_{jk}\in\mathbb{Z}}{\sum}x_{jk}MT  p^k_{j1}  +\underset{x_{jk}\in\mathbb{Z}}{\sum} \left(x_{jk}MT+\frac{M}{2}T\right)  p^k_{j2},
\end{align}
where $p^k_{j1}=\mathbb{P}(y_{jk}=x_{jk}MT|E^k_{jm})$ and  $p^k_{j2}=\mathbb{P}(y_{jk}=x_{jk}MT+\frac{M}{2}T|E^k_{jm})$, for $k\in \{m, m'\}$. 

 By using the fact that the users' channel gains in different time slots are i.i.d.,  the conditional expectation, 
$\mathcal{E}\{Y_k|E^k_{jm}\} $, can be rewritten  as follows:
\begin{align}
\mathcal{E}\{Y_k|E^k_{jm}\} \overset{(1)}{=}&  \sum^{\infty}_{j=1} jMT p_{0k}^{j-1}p_{mk} + \sum^{\infty}_{j=1} \left(jMT+\frac{M}{2}T\right)p_{0k}^{j-1}p_{m'k}  
 \\\nonumber 
 \overset{(2)}{=}&    \frac{MT (p_{mk} +p_{m'k} )}{(1- p_{0k})^2} + \frac{1}{2} \frac{MTp_{m'k}}{1- p_{0k}},
\end{align}
where the first step follows from 
 $p_{j1}^k=p_{0k}^{j-1}p_{mk}$ and  $p_{j2}^k=p_{0k}^{j-1}p_{m'k}$,   and 
 the last step follows by using \eqref{sum1} and \eqref{sum2}.

On the other hand, conditioned on $ E^k_{jm'}$, $y_{jk} $ can have the following two forms:  
\begin{align} \label{atwill yj3}
y_{jk} = \left\{\begin{array}{ll}\hspace{-0.5em}
x_{jk}MT,& \text{from   $t_i^{m'+1}$   to $x_{jk}MT + t_i^{m'+1}$ }  
\\\hspace{-0.5em}
x_{jk}MT-\frac{M}{2}T,& \text{from  $t_i^{m'+1}$   to $x_{jk}MT -\frac{M}{2}T+t_i^{m'+1}$ } 
\end{array}\right.\hspace{-1em}.
\end{align}
 Based on the above options for  $y_{jk} $,  the conditional expectation, $\mathcal{E}\{Y_k|E^k_{jm'}\} $, can be obtained as follows:
 \begin{align}
\mathcal{E}\{Y_k|E^k_{jm'}\} =& \underset{x_{jk}\in\mathbb{Z}}{\sum}x_{jk}MT  p^k_{j3}    +\underset{x_{jk}\in\mathbb{Z}}{\sum} \left(x_{jk}MT-\frac{M}{2}T\right)  p^k_{j4} \\\nonumber
 =& \sum^{\infty}_{j=1} jMT p_{0k}^{j-1}p_{m'k} + \sum^{\infty}_{j=1} \left(jMT-\frac{M}{2}T\right)p_{0k}^{j-1}p_{mk}  
 \\\nonumber 
 =&   \frac{MT (p_{mk} +p_{m'k} ) }{(1- p_{0k})^2}- \frac{1}{2} \frac{MTp_{mk}}{1- p_{0k}},
\end{align}
where    $p_{j3}^k\triangleq \mathbb{P}(y_{jk}=x_{jk}MT|E^k_{jm'})=p_{0k}^{j-1}p_{m'k}$,    and   $p^k_{j4}\triangleq \mathbb{P}(y_{jk}=x_{jk}MT-\frac{M}{2}T|E^k_{jm'})=p_{0k}^{j-1}p_{mk}$.

By using the two conditional expectations, $\Delta_{k,0}  $ can be obtained as follows:
\begin{align}\nonumber
\Delta_{k,0}  
\overset{(1)}{=}&
 \frac{mT\mathbb{P}(E^k_{jm})\mathcal{E}\{Y_k|E^k_{jm}\}   + m'T\mathbb{P}(E^k_{jm'})\mathcal{E}\{Y_k|E^k_{jm'}\}  }{\mathcal{E}\{Y_k\}} 
 \\\nonumber
 \overset{(2)}{=}&
 \frac{mTp_{mk}\mathcal{E}\{Y_k|E^k_{jm}\}   + m'Tp_{m'k}\mathcal{E}\{Y_k|E^k_{jm'}\}  }{\mathcal{E}\{Y_k\}} 
 \\ 
 \overset{(3)}{=}&\frac{1}{(p_{mk}+p_{m'k})^2  \frac{1}{(1-p_{0k})^2}}  \left[ \left(  (p_{mk} +p_{m'k} )  \frac{mTp_{mk}}{(1- p_{0k})^2} + \frac{p_{m'k}}{2} \frac{mTp_{mk}}{1- p_{0k}}\right)\right. \\\nonumber &+ \left.  \left( (p_{mk} +p_{m'k} )  \frac{m'Tp_{m'k}}{(1- p_{0k})^2} - \frac{p_{mk}}{2} \frac{m'Tp_{m'k}}{1- p_{0k}}\right)  \right],
\end{align}
where the second step follows from the fact that $\mathbb{P}(E^k_{jm})=p_{mk}$, $\mathbb{P}(E^k_{jm'})=p_{m'k}$, and the last step follows from the fact that $\mathcal{E}\{Y_k\}=MT(p_{mk}+p_{m'k})^2  \frac{1}{(1-p_{0k})^2} $. As can be seen from the above expression, the first term of the users'   AoI  expression in \eqref{delta request},  $\bar{\Delta}_k^N$, can be explicitly written  as a function of $p_{0k}$, $p_{mk}$, and $p_{m'k}$, which will be evaluated in the following two subsections for the two users, respectively.

\subsection{Evaluation of $p_{0m}$, $p_{mm}$, and $p_{m'm}$}
Recall that  $p_{0m}$ is the probability of the event that ${\rm U}_m$ fails to deliver an update in both of the two  time slots of  a given  frame.  Note that in each of the two times slots, ${\rm U}_m$ is allowed to transmit in different roles.   Further note that if the update from ${\rm U}_{m'}$ in the $m$-th time slot is successful, ${\rm U}_{m'}$ will remain silent in the $m'$-th time slot, which means that ${\rm U}_{m}$ solely occupies this time slot. By using the date rate constraint in \eqref{crnoma}, the probability, $p_0$, can  be expressed as follows: 
{\small \begin{align}\nonumber
p_{0m} =&   \mathbb{P}\left( \log\left( 1+ P|h_m^{i,m}|^2 \right)\leq \frac{N}{T},  \log\left( 1+\frac{P^S|h_{m'}^{i,m}|^2}{P|h_{m}^{i,m}|^2+1}\right)\leq \frac{N}{T}, \log\left( 1+\frac{P^S|h_m^{i,m'}|^2}{P|h_{m'}^{i,m'}|^2+1}\right)\leq \frac{N}{T} \right) \\\label{p0}
&+
 \mathbb{P}\left( \log\left( 1+ P|h_m^{i,m}|^2 \right)\leq \frac{N}{T},  \log\left( 1+\frac{P^S|h_{m'}^{i,m}|^2}{P|h_{m}^{i,m}|^2+1}\right)\geq \frac{N}{T},  \log\left( 1+ P^S|h_m^{i,m'}|^2 \right)\leq \frac{N}{T} \right) ,
\end{align}  }
$\hspace{-0.5em}$where   the $i$-th frame is used for illustration. 

Because the users' channel gains in different time slots are assumed to be independent, the event $E_1\triangleq \left\{\log\left( 1+ P|h_m^{i,m}|^2 \right)\leq \frac{N}{T},  \log\left( 1+\frac{P^S|h_{m'}^{i,m}|^2}{P|h_{m}^{i,m}|^2+1}\right)\leq \frac{N}{T}\right\}$ is independent from the event $\left\{\log\left( 1+\frac{P^S|h_m^{i,m'}|^2}{P|h_{m'}^{i,m'}|^2+1}\right)\leq \frac{N}{T}\right\}$. Therefore, the probability of $E_1$ can be calculated separately  as follows: 
\begin{align}\label{E1}
  \mathbb{P}\left( E_1 \right)  
=  & \mathbb{P}\left(  P|h_m^{i,m}|^2 \leq \epsilon,   \frac{P^S|h_{m'}^{i,m}|^2}{P|h_{m}^{i,m}|^2+1} \leq \epsilon  \right)  
\\\nonumber
=  & \int^{\frac{\epsilon}{P}}_{0} \left(1-e^{- \frac{\epsilon}{P^S}(Px+1)}\right)e^{-x}dx
=  1-e^{-\frac{\epsilon}{P}}-  e^{- \frac{\epsilon}{P^S}}\frac{1-e^{-\left(\frac{\epsilon}{P^S}P+1\right)\frac{\epsilon}{P} }}{\frac{\epsilon}{P^S}P+1}.
\end{align}
Similarly define $E_2\triangleq \left\{ \log\left( 1+ P|h_m^{i,m}|^2 \right)\leq \frac{N}{T},  \log\left( 1+\frac{P^S|h_{m'}^{i,m}|^2}{P|h_{m}^{i,m}|^2+1}\right)\geq \frac{N}{T}\right\}$. The probability of $E_2$ can be evaluated  as follows:
\begin{align}\label{E2}
  \mathbb{P}\left(E_2 \right)  
=  & \mathbb{P}\left(   P|h_m^{i,m}|^2 \leq \epsilon,   \frac{P^S|h_{m'}^{i,m}|^2}{P|h_{m}^{i,m}|^2+1} \geq \epsilon  \right)    
\\\nonumber
=  &   e^{- \frac{\epsilon}{P^S} }\frac{1-e^{-(\frac{\epsilon}{P^S}P +1)\frac{\epsilon}{P} }}{\frac{\epsilon}{P^S}P +1}.
\end{align}

By substituting \eqref{E1} and \eqref{E2} into \eqref{p0} and with some algebraic manipulations,    probability $p_{0m}$ can be expressed as follows:
\begin{align}\nonumber
p_{0m}
=& \left(1-e^{-\frac{\epsilon}{P}}-  e^{- \frac{\epsilon}{P^S}}\frac{1-e^{-\left(\frac{\epsilon}{P^S}P+1\right)\frac{\epsilon}{P} }}{\frac{\epsilon}{P^S}P+1}\right)  \left(1- \frac{e^{-\frac{\epsilon}{P^S} }}{1+\frac{P\epsilon}{P^S}}\right)
 + e^{- \frac{\epsilon}{P^S} }\frac{1-e^{-(\frac{\epsilon}{P^S}P +1)\frac{\epsilon}{P} }}{\frac{\epsilon}{P^S}P +1}\left(1-e^{- \frac{\epsilon}{P^S}}\right)    .
\end{align}

 Recall that $p_{m'm}$ is the probability for the event that the user fails to    deliver its update in the $m$-th time slot   but successfully delivers the update in the $m'$-th time slot. This probability can be expressed as follows:
 \begin{align}\nonumber
p_{m'm} =  &  \mathbb{P}\left( \log\left( 1+ P|h_m^{i,m}|^2 \right)\leq \frac{N}{T},  \log\left( 1+\frac{P^S|h_{m'}^{i,m}|^2}{P|h_{m}^{i,m}|^2+1}\right)\leq \frac{N}{T}, \log\left( 1+\frac{P^S|h_m^{i,m'}|^2}{P|h_{m'}^{i,m'}|^2+1}\right)\geq \frac{N}{T} \right) \\\nonumber
&+ \mathbb{P}\left( \log\left( 1+ P|h_m^{i,m}|^2 \right)\leq \frac{N}{T},  \log\left( 1+\frac{P^S|h_{m'}^{i,m}|^2}{P|h_{m}^{i,m}|^2+1}\right)\geq \frac{N}{T}, \log\left( 1+ P^S|h_m^{i,m'}|^2 \right)\geq \frac{N}{T} \right)  \\\nonumber
=& \left(1-e^{-\frac{\epsilon}{P}}-  e^{- \frac{\epsilon}{P^S}}\frac{1-e^{-\left(\frac{\epsilon}{P^S}P+1\right)\frac{\epsilon}{P} }}{\frac{\epsilon}{P^S}P+1}\right)   e^{-\frac{\epsilon}{P^S} }\frac{1}{1+\frac{P\epsilon}{P^S}} 
+ e^{- \frac{2\epsilon}{P^S} }\frac{1-e^{-(\frac{\epsilon}{P^S}P +1)\frac{\epsilon}{P} }}{\frac{\epsilon}{P^S}P +1}  ,
\end{align}
  where the last step is obtained by following   steps similar to those for evaluating  $p_{0m}$.   It is straightforward to show that  $p_{mm}=p_m$.

By substituting  the expressions of $p_{0m}$, $p_{mm}$, $p_{m'm}$ and $\Delta_{m,0}$ in \eqref{delta request}, a closed-form expression for  ${\rm U}_{m}$'s AoI can be obtained as shown in the lemma.

\subsection{Evaluation of  $p_{0m'}$, $p_{mm'}$ and $p_{m'm'}$}
Recall that $ {p}_{0m'}$ is the probability of the event that ${\rm U}_{m'}$ fails to deliver an update in  one frame. Again,  we take the $i$-th frame as an example. Note that  ${\rm U}_{m'}$ is the secondary user in the $m$-th time slot and the primary user in the $m'$-th time slot, which means that $p_{0m'} $ can be expressed as follows: 
\begin{align}
p_{0m'} =  & \mathbb{P}\left(  \log\left( 1+\frac{P^S|h_{m'}^{i,m}|^2}{P|h_{m}^{i,m}|^2+1}\right)\leq \frac{N}{T},  \log\left( 1+ P|h_{m'}^{i,m'}|^2 \right)\leq \frac{N}{T}\right) .
\end{align}
It is interesting to observe that the expression for $p_{0m'} $ is simpler than that for $p_{0m}$   in \eqref{p0} because ${\rm U}_{m'}$ is the primary user in the $m'$-th time slot and its data rate in this time slot is always $\log\left( 1+ P|h_{m'}^{i,m'}|^2 \right)$, regardless of  ${\rm U}_m$'s transmission strategy in the $m'$-th time slot.   

By using the assumption that the users' channels are i.i.d. Rayleigh faded, $p_{0m'}$ can be evaluated as follows:
\begin{align}
p_{0m'} =  & \mathbb{P}\left(   \frac{P^S|h_{m'}^{i,m}|^2}{P|h_{m}^{i,m}|^2+1} \leq \epsilon\right)\mathbb{P}\left(   |h_{m'}^{i,m'}|^2  \leq \frac{\epsilon}{P}\right) \\\nonumber 
=  & \left(1-e^{-\frac{\epsilon}{P^S}}\frac{1}{1+\frac{\epsilon P}{P^S}}\right)(1-e^{-\frac{\epsilon}{P}}). 
\end{align} 

Recall that $p_{mm'}$ denotes the probability of the event that ${\rm U}_{m'}$   successfully delivers its update in the $m$-th time slot of a frame, which  can be  expressed as follows:
\begin{align}
p_{mm'} =  & \mathbb{P}\left(   \log\left( 1+\frac{P^S|h_{m'}^{i,m}|^2}{P|h_{m}^{i,m}|^2+1}\right) \geq \frac{N}{T}  \right)  \\\nonumber &=e^{-\frac{\epsilon}{P^S}}\frac{1}{1+\frac{\epsilon P}{P^S}}. 
\end{align}
Recall that $p_{m'm'}$ denotes  the probability of the event that ${\rm U}_{m'}$ fails   in the $m$-th time slot  but successfully delivers the update in the $m'$-th time slot. This probability can be expressed as follows:
\begin{align}
p_{m'm'} =  & \mathbb{P}\left(  \log\left( 1+\frac{P^S|h_{m'}^{i,m}|^2}{P|h_{m}^{i,m}|^2+1}\right)\leq \frac{N}{T},  \log\left( 1+ P|h_{m'}^{i,m'}|^2 \right)\geq \frac{N}{T}\right) 
\\\nonumber 
=&\left(1-e^{-\frac{\epsilon}{P^S}}\frac{1}{1+\frac{\epsilon P}{P^S}}\right)e^{-\frac{\epsilon}{P}}.
\end{align}

By using the expressions for $p_{0m'}$, $p_{mm'}$, and $p_{m'm'}$, a closed-form expression for  ${\rm U}_{m'}$'s AoI can be obtained as shown in the lemma.  This completes the proof. 
\bibliographystyle{IEEEtran}
\bibliography{IEEEfull,trasfer}
  \end{document}